\documentclass[a4paper,10pt]{article}
\usepackage[utf8x]{inputenc}
\usepackage{amsmath}
\usepackage{times}
\usepackage{graphicx}
\usepackage{color}
\usepackage{multirow}
\usepackage{rotating}
\usepackage{bbm}
\usepackage{latexsym}
\usepackage[top=3cm, bottom=3cm, left=2.5cm, right=2.5cm]{geometry}

\usepackage{amsfonts}
\usepackage{mathrsfs}
\usepackage{amssymb}
\usepackage{amsthm}
\usepackage{subfigure}
\usepackage{subeqnarray}
\usepackage{cases}
\usepackage{float}
\usepackage{colortbl}
\usepackage{textcomp}
\usepackage{xcolor}

\theoremstyle{plain}
\theoremstyle{remark}
\newtheorem{remark}{Remark}[section]
\newtheorem{proposition}{Proposition}[section]
\newtheorem{corollary}{Corollary}[section]

\graphicspath{{Figures/}}

\title{A theoretical study of the role of astrocyte activity in neuronal hyperexcitability using a new neuro-glial mass model}
\author{Aurélie Garnier\,\footnote{Sorbonne Universités, UPMC Univ Paris 06, CNRS, INSERM, Laboratoire d'Imagerie Biomédicale (LIB), F-75013 Paris, France}\, \footnotemark[3] \,, Alexandre Vidal \footnote{Laboratoire de Mathématiques et Modélisation d'Évry (LaMME), CNRS UMR 8071, Université d'Évry-Val-d'Essonne, F-91000 Évry, France}\,\,, Habib Benali \footnotemark[1]}

\begin{document}

\maketitle

\renewcommand{\thefootnote}{\fnsymbol{footnote}}
\footnotetext[3]{aurelie.garnier@lib.upmc.fr}

\vspace{-0.5cm}

\begin{abstract}
The investigation of the neuronal environment allows us to better understand the activity of a cerebral region as a whole. The recent experimental evidences of the presence of transporters for glutamate and GABA in both neuronal and astrocyte compartments raise the question of the functional importance of the astrocytes in the regulation of the neuronal activity.  We propose a new computational model at the mesoscopic scale embedding the recent knowledge on the physiology of neuron and astrocyte coupled activities. The neural compartment is a neural mass model with double excitatory feedback, and the glial compartment focus on the dynamics of glutamate and GABA concentrations. Using the proposed model, we first study the impact of a deficiency in the reuptake of GABA by astrocytes, which implies an increase in GABA concentration in the extracellular space. A decrease in the frequency of neural activity is observed and explained from the dynamics analysis. Second, we investigate the neuronal response to a deficiency in the reuptake of Glutamate by the astrocytes. In this case, we identify three behaviors : the neural activity may either be reduced, or enhanced or, alternatively, may experience a transient of high activity before stabilizing around a new activity regime with a frequency close to the nominal one. After translating theoretically  the neuronal excitability modulation using the bifurcation structure of the neural mass model, we state the conditions on the glial feedback parameters corresponding to each behavior. \\

\noindent {\bf Keywords:} Model in Neuroscience, Qualitative analysis of dynamical systems, Bifurcations, Neuro-glial interactions, GABAergic and glutamatergic neurotransmissions, Excitability modulation, Neuronal hyperexcitability \\

\noindent {\bf AMS classification:} 
 34C15, 
 34C23, 
 34C28, 
 34C46, 
 34C60, 
 34H20, 
 92C20, 
 92B25. 
\end{abstract}

\section{Introduction}

For several years, to understand the mechanisms of the cerebral metabolism has become an important issue in neuroscience. The investigation of the neural environment allows us to better understand the activity of a cerebral region as a whole. For instance, the neural activity is composed of an interplay between excitation and inhibition where the cerebral blood flow (CBF) dynamics is an essential element as it reflects nutriments supplies such as oxygen and glucose. Synaptic transmission \cite{Araque_1998, Araque_1999, Halassa_2007, Nadkarni_2008, Zhang_2003} and neural activity \cite{Giaume_2010, Halassa_2010, Parpura_2000} are regulated by neurotransmitters, ions and molecules. Broadly speaking, at the microscopic scale, the presynaptic neuron releases neurotransmitters in the synaptic cleft which may bind to postsynaptic neuron receptors. These receptors, when ``activated'' by neurotransmitters, trigger ions exchanges between the extracellular space and the postsynaptic neuron that can induce the activation of the postsynaptic neuron. Parallely, neurotransmitter in the extracellular space may bind to astrocytic receptors and transporters inducing the activation of the glial calcic cycle allowing to release glutamate in the extra-synaptic space.

The demonstration of the presence of transporters for glutamate and GABA in both neural and astrocyte compartments raises the question of the functional importance of the astrocytes in the regulation of the neural activity \cite{Wang_2008}. It has been shown that glial cells and particularly astrocytes have a great impact on both the metabolic regulation \cite{Schousboe_2013} and the CBF dynamics \cite{Kowianski_2013}. Indeed, astrocytes modulate the dynamics of neurotransmitter concentrations, and thus the neuronal excitability, the synaptic transmission and the neural activity. Consequently, astrocytes dysfunctions are involved in several brain pathologies \cite{Losi_2012, Pittenger_2011, Seifert_2006}. Studying the interactions between the neurons and astrocytes cells has therefore become an essential problem in neurophysiology and biophysics.

In this context, computational models provide a key tool for interpreting the observed electrophysiological data and for revealing the different (patho-)physiological mechanisms which may underlie the observational data. Several models including the metabolic regulation mechanisms have been proposed in the literature. Some authors have built models of tripartite synapse \cite{Nadkarni_2007, Postnov_2009, Postnov_2007, Volman_2012, Volman_2007} considering an astrocyte coupled with a presynaptic neuron and a postsynaptic one and the dynamics of neurotransmitters, ions and molecules. Other authors considered a single neuron coupled with an astrocyte \cite{DePitta_2011, Gruetter_2001, Silchenko_2008, Volman_2012} and sometimes with a hemodynamic compartment \cite{Aubert_2002, Aubert_2005b, Aubert_2005} and the dynamics of neurotransmitters, ions and molecules between and in these elements. Recently, a new mesoscopic computational model \cite{Blanchard_2015} focusing on the astrocyte dynamics has been proposed. This model links the neural activity measured by local field potentials (LFP) to the CBF dynamics measured by Laser Doppler (LD) recordings through glial activity. The model incorporates the astrocyte cells via their role in neurotransmitters (glutamate and GABA) recycling, with physiologically-relevant relationships between these variables.

In this article, we propose a model extending the one presented in \cite{Blanchard_2015} for studying the neuro-glial interactions and more particularly the impact that the astrocyte activity may have upon the neuronal one. We use the same organization embedding a neural compartment and a glial one. To reproduce the mesoscopic neural activity, we use the neural mass model (NMM) with double excitatory feedback presented in \cite{Garnier_2015} which generalizes the Jansen-Rit model used in \cite{Blanchard_2015}. We keep the glutamate and GABA dynamics presented in \cite{Blanchard_2015} for the glial compartment. In this model, the neural activity acts on glutamate and GABA dynamics through the pyramidal and interneuronal firing rates respectively. The essential extension in our model consists in embedding the influence of the glial dynamics upon neural activity through the glutamate and GABA extracellular concentrations. Indeed, physiologically, a pyramidal cell (resp. an interneuron) releases glutamate (resp. GABA) in the synaptic cleft from where it binds to receptors on the postsynaptic neuron and the astrocyte. Reuptake processes (referred to as ``reuptakes'' in this article) of the neurotransmitters by the local astrocyte and presynaptic neuron regulates their concentration in the extracellular space (Figure \ref{Schema_NG_physio}). In the presynaptic neuron, the reuptake completes the stock whereas the reuptake by the astrocyte triggers a cascade of reactions linked with the modulation of synaptic transmission (differentially according to the type of neurotransmitters) and the hemodynamics. Hence, we introduce a feedback coupling from the glial compartment upon the neural one in our model. This feedback, referred to as the ``glial feedback'' in this article, allows us to study how different astrocyte deficiencies impact the local neuro-glial activities. If the mechanism of glutamate or GABA reuptake by the astrocytes is deficient, the neurotransmitter accumulates in the synaptic cleft, which leads to an increase in its concentration in the extracellular space. When this concentration reaches a threshold the synaptic transmission to postsynaptic neuron becomes abnormal and the postsynaptic neuron excitability threshold changes.

\begin{figure}[htbp]
 \centering
 \includegraphics[width=0.9\textwidth]{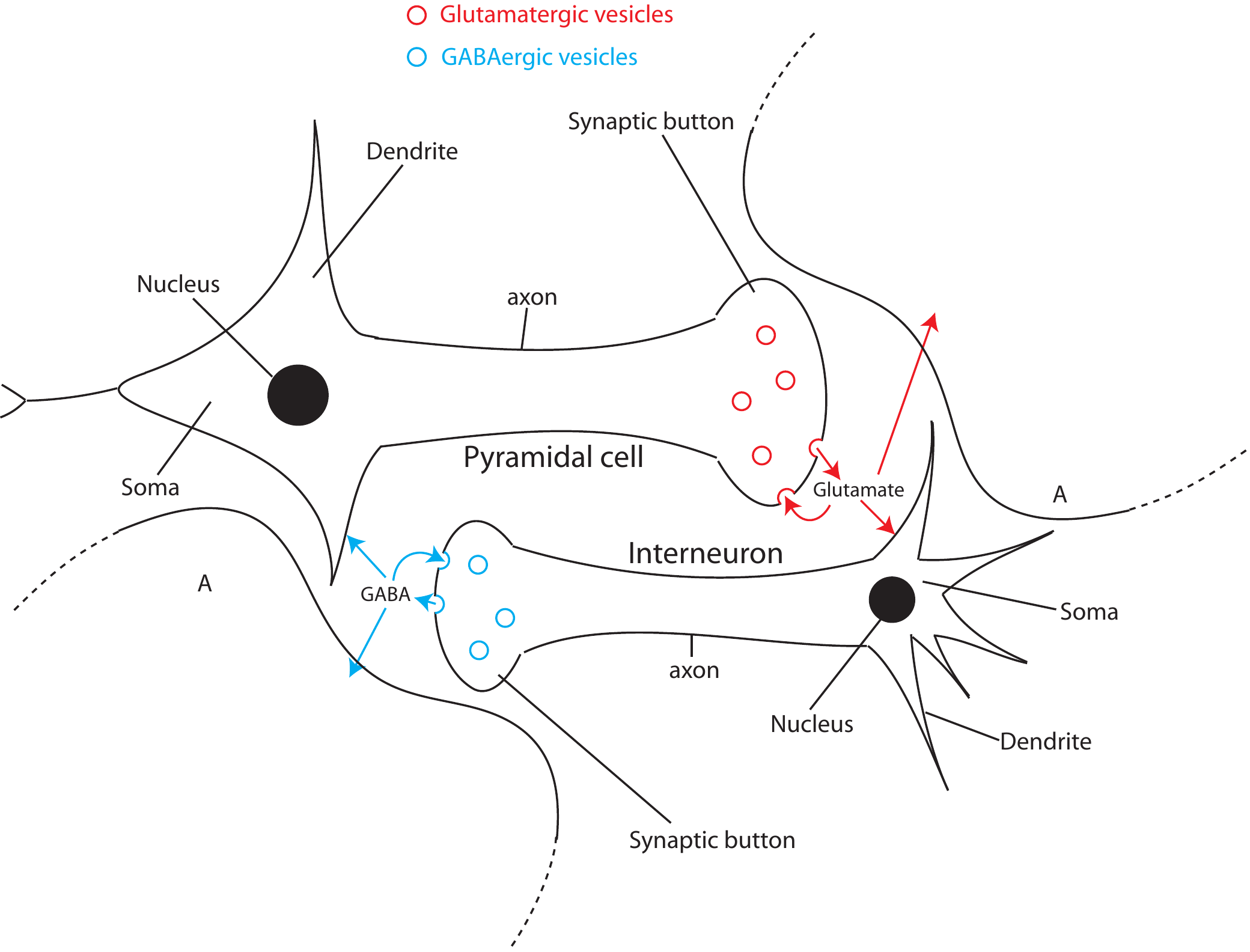}
 \caption{Scheme of neurotransmission mechanisms and neurotransmitter reuptake. Red circles: glutamatergic vesicles. Blue circles: GABAergic vesicles. Red arrows: exchanges of glutamate. Blue arrows: exchanges of GABA. A: astrocytes.}
 \label{Schema_NG_physio}
\end{figure}

The proposed model provides a unified framework in which knowledge of the physiology of neuron and astrocyte activities, as well as their couplings, can be incorporated, in order (i) to simulate the output signals for chosen parameter values, (ii) to identify the various qualitative responses of the whole system to physiological disorders, (iii) to study theoretically the conditions over the parameters corresponding to each type of response. The neural mass approach offers therefore an optimal compromise between the compactness of the dynamics, the richness of the physio-pathological mechanisms that can be reproduced, and the interpretability of the parameters from the biophysical viewpoint. The range of parameter values can be inspired from physiology-based studies in actual animal models. Once the parameter range is defined, one can test the influence of varying a reduced set of parameters (for example, ratio between excitation and inhibition or Glutamate and GABA reuptake) on the output signals.

The paper is organized as follows. In section 2, we discuss the main pathways of the neuro-glial interactions. We recall the dynamical features that explain the qualitative and quantitative properties of the time series generated by each compartment. In particular, we describe the bifurcation structure of the neural compartment underlying the generation of Noise Induced Spiking (NIS) outputs on which we focus in this article \cite{Garnier_2015}. We introduce the bilaterally coupled model and illustrate the main difference with the feedforward model by mimicking the injection of a GABA bolus in the extracellular space. In section 3, we study theoretically the effect of a deficiency of the GABA reuptake by the astrocyte upon the neural system behavior, we illustrate the result by numerical simulations and link the observed outputs with known biological results. In section 4, we identify and illustrate numerically the three possible responses of the neural compartment to a deficiency of the glutamate reuptake by the astrocyte : reduced activity, transient and permanent hyperexcitability. We explain this spectrum of responses using the analysis of the dynamical structure of the model and we derive explicit conditions on the parameters involved in the glial feedback corresponding to each type of behaviors. Finally we interpret the conditions on the parameters in terms of physio-pathological features and discuss possible applications of this study for investigating experimentally the precise role of astrocyte deficiencies in neuronal hyperexcitability.

\section{Neuro-glial mass approach : bilateral coupling of mesoscopic models}
\label{NG_pres}

In this section, we briefly recall the structure of the NMM generalizing the Jansen-Rit model and studied in \cite{Garnier_2015}, the way to analyze its properties and the predominant time series pattern that it generates. Then we describe the model introduced in \cite{Blanchard_2015} to reproduce the glial dynamics. Then we explain our choice of bilateral coupling between these two compartments. Finally, we illustrate its main dynamical differences with the feedforward coupling system ({\it i.e.} unilateral coupling from the neural compartment upon the glial one).

\subsection{Neural mass model and Noise-Induced Spiking} \label{ssec_NMM}
The NMM represents the dynamical interactions between two neural populations at a mesoscopic scale: a main population of pyramidal cells (P) and a population of inhibitory interneurons (I). It also involves the interactions of P with a general population P' representing neighboring pyramidal cells and interacting with P through synaptic connections. There are three feedback loops on population P activity: an inhibitory feedback through the interneuron population I, a direct excitatory feedback of P onto itself (referred to as ``direct feedback'') and an indirect excitatory feedback (referred to as ``indirect feedback'') involving the population P' (Figure \ref{NMM}(a)).

\begin{figure}[htbp]
 \centering
 \includegraphics[width=0.65\textwidth]{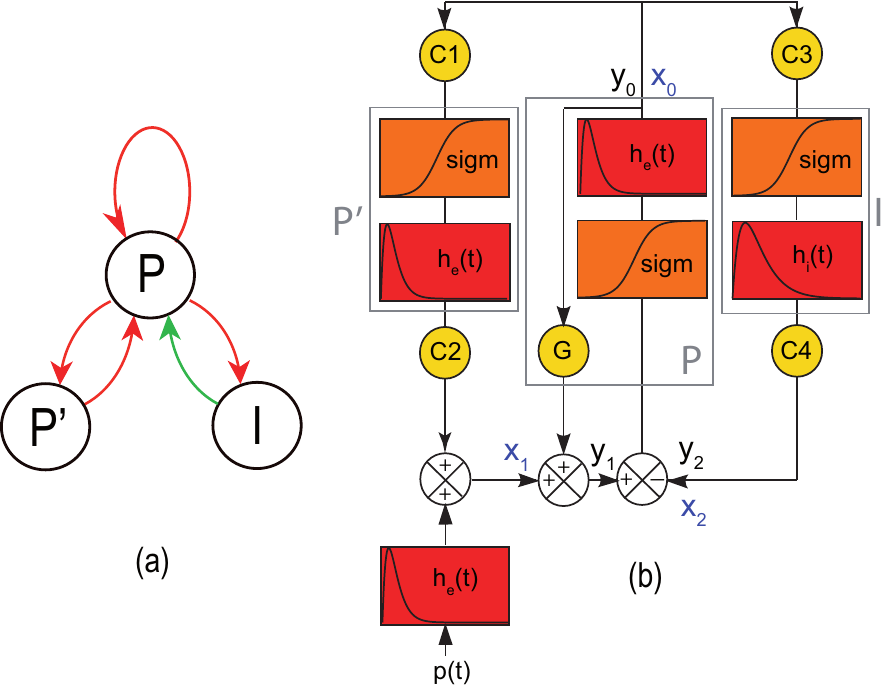}
 \caption{Two schematic representations of the NMM with double excitatory feedbacks. P: main population of pyramidal cells. I: Interneuron population. P': secondary population of pyramidal cells. Red (resp. green) arrows in (a): excitatory (resp. inhibitory) interactions. Box $h_e(t)$ (resp. $h_i(t)$): second order process converting action potentials into excitatory (resp. inhibitory) post-synaptic potential. Box ${\rm sigm}$: process converting average membrane potential into average action potential density discharge by neurons of populations P, P' and I respectively. $C_i$ for $i \in \text{\textlbrackdbl} 1,4 \text{\textrbrackdbl}$: coupling gain parameters depending on the maximal number $C$ of synaptic connections between two populations. $G$: direct feedback coupling gain. $p(t)$: excitatory input. $y_0$, $y_1$, $y_2$: state variables. $x_0$, $x_1$, $x_2$: intermediary variables.}
 \label{NMM}
\end{figure}

The conversion process of average pulse density into excitatory and inhibitory postsynaptic potential respectively are based on the following functions introduced by Van Rotterdam {\em et al.} \cite{VanRotterdam_1982}:
\begin{align*}
 h_e(t) &= A\,a\,t\,e^{-a\,t}, \\
 h_i(t) &= B\,b\,t\,e^{-b\,t}
\end{align*}
These functions are the basic solutions of the differential operators ${\cal F}_e$ and ${\cal F}_i$ respectively:
\begin{subequations} \label{Transfer_Dyn}
\begin{eqnarray}
 {\cal F}_e(h_e) &= \frac{1}{A} \left(\frac{1}{a}\,h_e''-2\,h_e'-a\,h_e \right) \\
 {\cal F}_i(h_e) &= \frac{1}{B} \left(\frac{1}{b}\,h_i''-2\,h_i'-b\,h_i \right)
  \end{eqnarray}
\end{subequations}
In this framework, parameter $A$ (resp. $B$) tunes the amplitude of excitatory (resp. inhibitory) postsynaptic potentials and $\frac{1}{a}$ (resp. $\frac{1}{b}$) represents the time constant of excitatory (resp. inhibitory) postsynaptic potentials representative of the kinetics of synaptic connections and delays introduced by circuitry of the dendritic tree \cite{Freeman_1975, VanRotterdam_1982, Jansen_1993}. Following Freeman's work \cite{Freeman_1975}, the sigmoidal functions converting the average membrane potential into an average pulse density have the following form:
\[
 {\rm sigm}(x,v) = \frac{2\,e_0}{1+e^{r\,(v-x)}}
\]
where $2\,e_0$ represents the maximum discharge rate, $v$ the excitability threshold and $r$ the sigmoid slope at the inflection point. Finally the NMM receives an excitatory input $p(t)$ standing for the action on population P of neural populations in other areas through long-range synaptic connections. Classically one consider $p(t)$ a gaussian variable to represent a non-specific input and generate the model outputs.

Now we can write the dynamics for the intermediary variables $x_0$, $x_1$ and $x_2$ which represent the outputs of the population P, the population P' and the population I respectively (Figure \ref{NMM}(b)):
\begin{subequations} \label{x_init}
\begin{align}
 x_0'' &= A\,a\,{\rm sigm}(x_1+G\,x_0-x_2,v_{\rm P})-2\,a\,x_0'-a^2\,x_0 \label{x0_init} \\
 x_1'' &= A\,a\,C_2\,{\rm sigm}(C_1\,x_0,v_{\rm P'})-2\,a\,x_1'-a^2\,x_1+A\,a\,p(t) \label{x1_init} \\
 x_2'' &= B\,b\,C_4\,{\rm sigm}(C_3\,x_0,v_{\rm I})-2\,b\,x_2'-b^2\,x_2 \label{x2_init}
\end{align}
\end{subequations}
Parameters $C_i$, $i\in  \text{\textlbrackdbl} 1,4 \text{\textrbrackdbl}$, represent the average number of synapses between two populations. Following \cite{Braitenberg_1998}, each $C_i$ is proportional to the maximal number $C$ of synapses between two populations. The excitation of P by its own output, resulting from the intra-population synaptic connections, is weighted by the coupling gain $G$.

For sake of comparison, we use a variable change to obtain the same state variables as in the Jansen-Rit model \cite{Jansen_1995} : the excitatory output ($y_0=x_0$) and the excitatory ($y_1=x_1+G\,x_0$) and inhibitory ($y_2=x_2$) inputs of the main population P. The output $y_0$ acts on the secondary pyramidal cell population P' and on the interneuron population I.
To analyze the model, we write the dynamics of the state variables $y_0$, $y_1$ and $y_2$ as a system of first order differential equations:
\begin{subequations} \label{y_init}
\begin{align}
 y_0' &= y_3 \\
 y_1' &= y_4 \\
 y_2' &= y_5 \\
 y_3' &= A\,a\,{\rm sigm}(y_1-y_2,v_{\rm P})-2\,a\,y_3-a^2\,y_0 \label{y0_init} \\
 y_4' &= A\,a\,C_2\,{\rm sigm}(C_1\,y_0,v_{\rm P'})+A\,a\,G\,{\rm sigm}(y_1-y_2,v_{\rm P})-2\,a\,y_4-a^2\,y_1+A\,a\,p(t) \label{y1_init} \\
 y_5' &= B\,b\,C_4\,{\rm sigm}(C_3\,y_0,v_{\rm I})-2\,b\,y_5-b^2\,y_2 \label{y2_init}
\end{align}
\end{subequations}
In this article, we consider the local field potential (LFP) as the main model output. Following \cite{Jansen_1993}, it is defined by ${\rm LFP}(t)=y_1(t)-y_2(t)$. It is important to note that, generally, studies of neural mass models, such as Jansen-Rit model, only considered the case with the same constant excitability thresholds for all populations, {\it i.e.}
\[
 v_{\rm P}=v_{\rm P'}=v_{\rm I}=v_0
\]

The behavior of NMMs can be deduced from the bifurcation diagram according to the value $p(t)=p$ considered as a parameter, as it has been performed in \cite{Touboul_2011} on the Jansen-Rit model. In \cite{Garnier_2015}, we have classified the types of time series patterns generated by model \eqref{y_init} and the associated bifurcation structures according to the strengths of the different excitative feedbacks applied to population $P$.
Let us recall the bifurcation diagram underlying the predominant type of generated time series, which we will consider in this article.

Model \eqref{y_init} has the following useful particularities that have been highlighted in \cite{Garnier_2015}. First, for a fixed value of parameter $p$, the $y_0$ value of a singular point suffices to have explicit expressions of all the other components. Second, for a given $y_0$ value, there exists a unique value of $p$ such that $y_0$ corresponds to a singular point. In other terms, the set of singular points obtained for the different values of $p$ is a graph over $y_0$. Hence, we can visualize the shape of the singular point locus in the plane $(p,y_0)$ : in the case presented here (see Figure \ref{NIS_DB}), this curve of singular points is S-shaped. In the following description, for a given bifurcation ``bif'' according to parameter $p$, we note $p_{bif}$ the bifurcation value and, if the bifurcation involves a singular point, we note $y_{bif}$ the corresponding $y_0$ value.

\begin{figure}[htbp]
 \centering
 \includegraphics[width=0.95\textwidth]{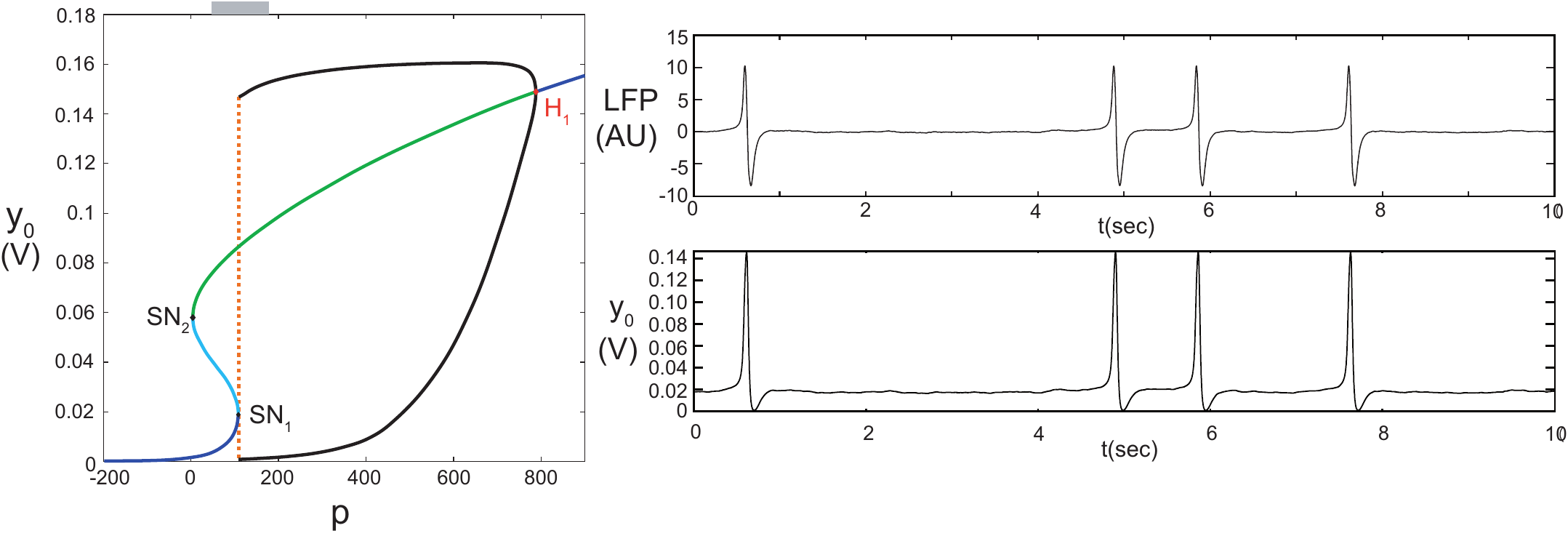}
 \caption{Bifurcation diagram according to $p$ (left) and associated LFP and $y_0$ time series (right). Blue curves: stable singular points. Cyan (resp. green) curves: singular points with one (resp. two) eigenvalues with positive real parts. Black curves: $y_0$ extrema along stable limit cycles. Black points (SN$_1$ and SN$_2$): saddle-node bifurcations. Red point (H$_1$): supercritical Hopf bifurcation. Dashed orange line: Saddle-Node on Invariant Circle (SNIC) bifurcation. Horizontal gray bar: confidence interval $[<p>-\sigma , <p>+\sigma ]$ of the gaussian variable $p(t)$ used to generate the time series.}
 \label{NIS_DB}
\end{figure}

Two saddle-node bifurcations SN$_1$ and SN$_2$ split the curve of singular points into three branches. We name ``lower branch'', ``middle branch'' and ``upper branch'' the sets of singular points satisfying $y_0<y_{\rm SN_1}$, $y_{\rm SN_1}<y_0<y_{\rm SN_2}$ and $y_0>y_{\rm SN_2}$ respectively. Singular points on the lower branch are stable (blue) and those on the middle branch are unstable (cyan). Singular points on the upper branch are unstable (green) for $p<p_{\rm H_1}$ and stable (blue) otherwise. At $p=p_{\rm H_1}$ the system undergoes a supercritical Hopf bifurcation H$_1$ giving birth to a stable limit cycle for $p<p_{\rm H_1}$ that persists until $p=p_{\rm SN_1}=p_{\rm SNIC}$ where it disappears by a saddle-node on invariant circle (SNIC) bifurcation (dashed orange line). The existence of the SNIC bifurcation is essential because it implies the appearance of a large amplitude stable limit cycle with  large period. Thereby, according to the value of $p$, the system alternates between oscillatory phases (for $p>p_{\rm SNIC}$) and quiescent phases (for $p<p_{\rm SNIC}$). In other terms, the value $p_{\rm SNIC}$ plays the role of an activation threshold for the neural compartment, which is a key point of the subsequent analysis.

Note that the oscillation frequency in the generated oscillatory pattern is driven by the value of $p$: as $p$ tends to $p_{\rm SNIC}$ from above, the limit cycle period tends to infinity. Hence, the closest $p$ is to $p_{\rm SNIC}$, the lowest is the frequency. Consequently, when considering a gaussian input for the model, the occurrence of spikes and their frequency depend on the features of the normal distribution generating $p(t)$, which led us to refer to the corresponding pattern as Noise-Induced Spiking (NIS) in \cite{Garnier_2015}.

Such LFP activity, {\it i.e.} sparse large amplitude spikes, corresponds to episodic synchronization of the neuron activities among the populations. Physiologically, this pattern of activity arise, as for interictal spiking activity, and is symptomatic of a strong excitability of the neuronal system that can turn into hyperexcitability during pathological crisis. For fixed parameters, the activity is stable in the sense that the oscillation frequency does not change much along time. In the following, we study the variations of the activity when the neural dynamics are altered by the surrounding activity, {\it i.e.} the glial feedback.

\subsection{Glial model : glutamate and GABA concentration dynamics}

For reproducing the glial activity, we use the model introduced in \cite{Blanchard_2015}. It focuses on the dynamics of glutamate and GABA concentrations, which are the main neurotransmitters of the central nervous system. In \cite{Blanchard_2015}, the neuro-glial coupling is feedforward: the glial dynamics is driven by the neural activity, generated by the Jansen-Rit model, but does not impact the neural compartment. The model considers the dynamics of glutamate and GABA concentrations, locally to the main population P of pyramidal cells, at different stages of the recycling mechanism. The local nature of this interaction implies that the firing rate of the secondary population P' of pyramidal cells does not impact the glial dynamics associated with the neighboring astrocytes of the main population P. The mechanism is as follows (Figure \ref{NG_sans_retroaction}): excited pyramidal cells (resp. interneurons) release glutamate (resp. GABA) in the extracellular space (synaptic cleft). Astrocytes and pre-synaptic neurons reuptake the neurotransmitters. Astrocytes recycle or consume the neurotransmitters while the presynaptic neurons capture them to complete their stock.

\begin{figure}[ht]
 \centering
 \includegraphics[width=0.6\textwidth]{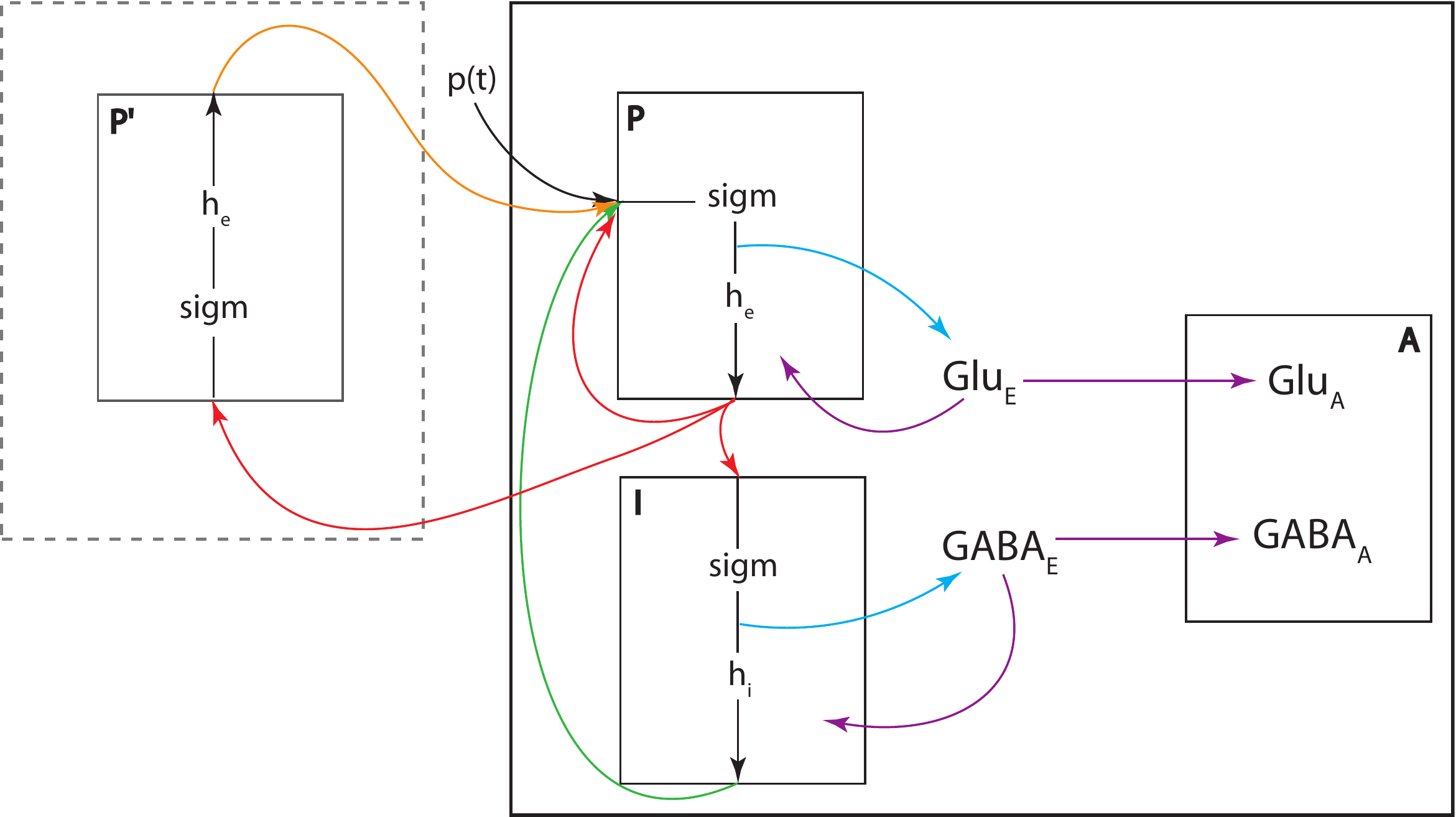}
 \caption{Scheme of the feedforward neuro-glial mass model. P and P': main and secondary populations of pyramidal cells. I: interneuron population. $p(t)$: excitatory input on population P. ${\rm [Glu]_{E}}$ and ${\rm [GABA]_{E}}$: glutamate and GABA extracellular concentration. ${\rm [Glu]_{A}}$ and ${\rm [GABA]_{A}}$: glutamate and GABA glial concentrations. Red arrows: P$\rightarrow$P, P$\rightarrow$I and P$\rightarrow$P' couplings. Orange arrow: P'$\rightarrow$P coupling. Green arrow: I$\rightarrow$P coupling. Cyan arrows: glutamate and GABA release by populations P and I into extracellular space. Purple arrows: glial and neural reuptakes of neurotransmitters.}
 \label{NG_sans_retroaction}
\end{figure}

Following \cite{Blanchard_2015} the glial compartment is built on the firing rate (${\rm FR}_{pyr}$) of the pyramidal cell population and the firing rate (${\rm FR}_{int}$) of the interneuron population. The state variables are
\begin{itemize}
\item ${\rm [Glu]_{NE}}$ and ${\rm [GABA]_{IE}}$ : the fluxes of glutamate and GABA from neurons to extracellular space,
\item  ${\rm [Glu]_{E}}$ and ${\rm [GABA]_{E}}$ : the neurotransmitter concentrations in the extracellular space,
\item  ${\rm [Glu]_{A}}$ and ${\rm [GABA]_{A}}$ : the quantity of neurotransmitters recycled and consumed by the astrocytes.
\end{itemize}
Naturally, the dynamics governing ${\rm [Glu]_{NE}}$ and ${\rm [GABA]_{IE}}$ are driven by second-order differential operators similar to the synaptic transfer dynamics introduced in \eqref{Transfer_Dyn} \cite{MArdekani_2013, VanRotterdam_1982}: 
\begin{align*}
 {\cal F}_{\rm Glu}(h_{\rm Glu}) &= \frac{1}{W}\left( \frac{1}{w_1}\,h_{\rm Glu}''-\frac{w_1+w_2}{w_1}\,h_{\rm Glu}'-w_2\,h_{\rm Glu} \right) \\
 {\cal F}_{\rm GABA}(h_{\rm GABA}) &= \frac{1}{Z} \left(\frac{1}{z_1}\,h_{\rm GABA}''-\frac{z_1+z_2}{z_1}\,h_{\rm GABA}'-z_2\,h_{\rm GABA} \right)
\end{align*}

As for the synaptic transfer functions, parameter $W$ (resp. $Z$) tunes the peak amplitude of glutamate (resp. GABA) concentrations and parameters $w_1$ and $w_2$ (resp. $z_1$ and $z_2$) tune the rise and decay times of glutamate (resp. GABA) release transfer function. These dynamics are well-suited for reproducing the qualitative and quantitative properties of rise and decay in neurotransmitter concentrations. 

The reuptakes of glutamate from the extracellular space by astrocyte (${\rm [Glu]_{EA}}$) and neurons (${\rm [Glu]_{EN}}$) are triggered when extracellular concentration of glutamate reaches a threshold. Moreover the efficiencies of these processes saturate for high concentrations values, which leads to model these dynamics using sigmoidal functions. GABA reuptakes (${\rm [GABA]_{EA}}$ and ${\rm [GABA]_{EI}}$) are modeled with Michaelis-Menten dynamics following the experimental literature \cite{Blanchard_2015}. The dynamics of the extracellular concentrations (${\rm [Glu]_{E}}$ and ${\rm [GABA]_{E}}$) are derived from the input and output fluxes described above. The astrocyte concentration dynamics (${\rm [Glu]_{A}}$ and ${\rm [GABA]_{A}}$) result from the glial reuptake ones and a linear consumption term.

In sigmoidal functions for ${\rm [Glu]_{E}}$ and ${\rm [GABA]_{E}}$, parameters $V_{\rm glu}^{\rm EA}$ and $V_{\rm glu}^{\rm EN}$ are the maximal velocities for glutamate reuptakes by the neurons and the astrocytes respectively, $s_g$ represents the activation threshold and $r_g$ the sigmoidal slope at the inflection point. Parameters $V_{\rm gba}^{\rm EA}$ and $K_{\rm gba}^{\rm EA}$ (resp. $V_{\rm gba}^{\rm EN}$ and $K_{\rm gba}^{\rm EN}$) are respectively the maximal velocity and concentration for glial (resp. neural) GABA transporter. Finally, $V_{\rm cglu}$ and $V_{\rm cgba}$ are the glutamate and GABA degradation rates in astrocytes. We refer the reader to \cite{Blanchard_2015} for a detailed explanation of the dynamics.

Hence, the feedforward model obtained by coupling the NMM defined by \eqref{y_init} and the glial dynamics introduced in \cite{Blanchard_2015} reads
\begin{subequations} \label{Sys_NG_SR}
\begin{align}
 y_0' &= y_3 \label{xx0}\\
 y_1' &= y_4 \\
 y_2' &= y_5 \\
 y_3' &= A\,a\,{\rm sigm}(y_1-y_2,v_{\rm P})-2\,a\,y_3-a^2\,y_0 \label{x0} \\
 y_4' &= A\,a\,C_2\,{\rm sigm}(C_1\,y_0,v_{\rm P'})+A\,a\,G\,{\rm sigm}(y_1-y_2,v_{\rm P})-2\,a\,y_4-a^2\,y_1+A\,a\,p(t) \label{x1} \\
 y_5' &= B\,b\,C_4\,{\rm sigm}(C_3\,y_0,v_{\rm I})-2\,b\,y_5-b^2\,y_2 \label{x2} \\
 {\rm [Glu]_{NE}}' &= {\rm d[Glu]_{NE}} \\
 {\rm d[Glu]_{NE}}' &= W\,w_1\,{\rm sigm}(y_1-y_2,v_{\rm P})-(w_1+w_2)\,{\rm d[Glu]_{NE}}-w_1\,w_2\,{\rm [Glu]_{NE}} \label{glune} \\
 {\rm [Glu]_{E}}' &= {\rm [Glu]_{NE}}-\frac{V_{\rm glu}^{\rm EA}}{1+e^{r_g\,s_g-r_g\,{\rm [Glu]_E}}}-\frac{V_{\rm glu}^{\rm EN}}{1+e^{r_g\,s_g-r_g\,{\rm [Glu]_E}}} \\
 {\rm [Glu]_A}' &= \frac{V_{\rm glu}^{\rm EA}}{1+e^{r_g\,s_g-r_g\,{\rm [Glu]_E}}}-V_{\rm cglu}\,{\rm [Glu]_A} \\
 {\rm [GABA]_{IE}}' &= {\rm d[GABA]_{IE}} \\
 {\rm d[GABA]_{IE}}' &= Z\,z_1\,{\rm sigm}(C_3\,y_0,v_{\rm I})-(z_1+z_2)\,{\rm d[GABA]_{IE}}-z_1\,z_2\,{\rm [GABA]_{IE}} \label{gabaie} \\
 {\rm [GABA]_E}' &= {\rm [GABA]_{IE}}-\frac{V_{\rm gba}^{\rm EA}}{K_{\rm gba}^{\rm EA}+{\rm [GABA]_E}}\,{\rm [GABA]_E}-\frac{V_{\rm gba}^{\rm EN}}{K_{\rm gba}^{\rm EN}+{\rm [GABA]_E}}\,{\rm [GABA]_E}\\
 {\rm [GABA]_A}' &= \frac{V_{\rm gba}^{\rm EA}}{K_{\rm gba}^{\rm EA}+{\rm [GABA]_E}}\,{\rm [GABA]_E}-V_{\rm cgba}\,{\rm [GABA]_A}
\end{align}
\end{subequations}

Table \ref{par_table} specifies the parameter values used for the simulations in the following sections. The values of parameters associated with the NMM have been chosen to reproduce NIS behavior using the analysis in \cite{Garnier_2015}.

\begin{table}[htb]
 \centering
 \begin{tabular}{|l|l|l|l|}
  \hline
  \multicolumn{2}{|c|}{Neurons} & \multicolumn{1}{c|}{Glutamate} & \multicolumn{1}{c|}{GABA} \\
  \hline
  $A=3.25\,{\rm mV}$ & $C=135$ & $W=53.6\,\mu {\rm M.s^{-1}}$ & $Z=53.6\,\mu {\rm M.s^{-1}}$ \\
  $B=22\,{\rm mV}$ & $\alpha_1=1$ & $w_1=90\,{\rm s}^{-1}$ & $z_1=90\,{\rm s}^{-1}$ \\
  $a=100\,{\rm s}^{-1}$ & $\alpha_2=0.8$ & $w_2=33\,{\rm s}^{-1}$ & $z_2=33\,{\rm s}^{-1}$ \\
  $b=50\,{\rm s}^{-1}$ & $\alpha_3=0.25$ & $V_{\rm glu}^{\rm EA}=4.5\,\mu {\rm M.s^{-1}}$ & $V_{\rm gba}^{\rm EA}=2\,\mu {\rm M.s^{-1}}$ \\
  $e_0=2.5\,{\rm s}^{-1}$ & $\alpha_4=0.25$ & $V_{\rm glu}^{\rm EN}=0.5\,\mu {\rm M.s^{-1}}$ & $K_{\rm gba}^{\rm EA}=8\,\mu {\rm M}$ \\
  $v_0=6\,{\rm mV}$ & $G=40$ & $r_g=0.9\,\mu {\rm M^{-1}}$ & $V_{\rm gba}^{\rm EN}=5\,\mu {\rm M.s^{-1}}$ \\
  $r=0.56\,{\rm mV}^{-1}$ &  & $s_g=6\,\mu {\rm M}$ & $K_{\rm gba}^{\rm EN}=24\,\mu {\rm M}$ \\
   &  & $V_{\rm cglu}=9\,\mu {\rm M.s^{-1}}$ & $V_{\rm cgba}=9\,\mu {\rm M.s^{-1}}$ \\
 \hline
 \end{tabular}
\caption{Values of the neuro-glial model parameters.}
\label{par_table}
\end{table}

System \eqref{Sys_NG_SR} is built as a feedforward coupling of the neural compartment onto the glial one. Hence, in this model, the neural compartment is not impacted by the neurotransmitter concentrations in the extracellular space. As mentioned in the introduction, these concentrations have been proven to modulate the local neuron excitability and this feedback has been identified in recent studies \cite{Araque_1998} to be an essential mechanism of several pathologies triggered by glial reuptake deficiencies. Consequently, our aim is to include such feedback in the model in order to study the effects of different astrocyte dysfunctioning on the neuronal activity.

\subsection{Glial feedback and neuro-glial mass model} \label{NG_AR}
The concentrations of neurotransmitters in a synaptic cleft act on the excitability threshold of the post-synaptic neuron. In the neuro-glial model \eqref{Sys_NG_SR} the alteration of this neural excitability threshold can be reproduced by dynamical changes in $v_P$, $v_{P'}$ and $v_I$. In the following, we describe how we model the modulation of the neuron excitability in each population by the neurotransmitter concentrations in the extracellular space basing ourselves on biological knowledge.

Extracellular concentrations of neurotransmitters have a thresholded impact on neural activity \cite{Araque_1998}. Precisely, on one hand, the concentrations must be large enough to impact significantly the neural activity. On the other hand, the postsynaptic neurons are saturated when these concentrations become to large and, consequently, the neural excitability remains bounded. It is worth noticing that quantitative experimental data of the impact of neurotransmitter concentrations on neural excitability do not exist up to now. By default, we consider sigmoidal functions to model the glial feedback on neural excitability which is a natural choice for aggregating the qualitative experimental knowledge. Yet the upcoming mathematical analysis can be easily extended to any bounded increasing functions with a unique inflection point.

We introduce three sigmoidal functions to model the components of the glial feedback: ${\rm Si_{Glup}}$ for the glutamate feedback on pyramidal cells, ${\rm Si_{Glui}}$ for the glutamate feedback on interneurons and ${\rm Si_{GABA}}$ for the GABA feedback on pyramidal cells. We parameterize these functions as follows
\begin{subequations} \label{sigm_fdb}
\begin{align}
 {\rm Si_{Glup}}({\rm [Glu]_E}) &= \frac{m_{\rm Glup}}{1+e^{r_{\rm Glup}\,(v_{\rm Glup}-{\rm [Glu]_E})}} \\
 {\rm Si_{Glui}}({\rm [Glu]_E}) &= \frac{m_{\rm Glui}}{1+e^{r_{\rm Glui}\,(v_{\rm Glui}-{\rm [Glu]_E})}} \\
 {\rm Si_{GABA}}({\rm [GABA]_E}) &= \frac{m_{\rm GABA}}{1+e^{r_{\rm GABA}\,(v_{\rm GABA}-{\rm [GABA]_E})}}
\end{align}
\end{subequations}
The parameter values used for the simulations in the following sections are given in Table \ref{par_feedback} and have been chosen to reproduce an average physiological behavior.

\begin{table}[htbp]
 \centering
 \begin{tabular}{|c|c|c|}
 \hline
 $m_{\rm Glup}=2.5$ & $m_{\rm Glui}=1$ & $m_{\rm GABA}=1$ \\
 $r_{\rm Glup}=0.15$ & $r_{\rm Glui}=0.15$ & $r_{\rm GABA}=0.12$ \\
 $v_{\rm Glup}=30$ & $v_{\rm Glui}=30$ & $v_{\rm GABA}=25$ \\
 \hline
 \end{tabular}
\caption{Parameter values of the sigmoidal function ${\rm Si_{Glui}}({\rm [Glu]_E})$, ${\rm Si_{Glup}}({\rm [Glu]_E})$ and ${\rm Si_{GABA}}({\rm [GABA]_E})$.}
\label{par_feedback}
\end{table}

We note that the fixation mechanisms of glutamate on pyramidal cells and interneurons are the same since the neurotransmitter transporters are independent on the type of neuron. Thus, only parameters $m_{\rm Glup}$ and $m_{\rm Glui}$ representing the maximal coupling gains of the glutamate-related component of the glial feedback discriminate between the coupling functions ${\rm Si_{Glup}}$ and ${\rm Si_{Glui}}$, since the synaptic sensitivities may not be the same in pyramidal cells and interneurons.

At the beginning of this subsection, we evoked that the glial feedback acts on the excitability thresholds of neurons. More specifically, if there is an excess of neurotransmitter in a synapse from a neuron ${\rm n}_1$ of population ${\rm N}_1$ to a neuron ${\rm n}_2$ of population ${\rm N}_2$, the extracellular concentration of neurotransmitter acts on the postsynaptic neuron ${\rm n}_2$ by changing its excitability threshold. In system \eqref{Sys_NG_SR} the excitability threshold of neurons, that is a parameter at the individual scale, does not appear explicitly. However, when the excitability of the population ${\rm N}_2$ neurons changes at the individual scale, the number of neurons activated in this population by a given input changes as well and thus the output of this population is also modified. Consequently, we choose to change parameter $v_{{\rm N}_2}$ in the equation corresponding to the output of population ${\rm N}_2$ since this parameter represents a modulation of the threshold of the sigmoidal function ${\rm sigm}$.

Let us now describe how we build the feedbacks on the dynamics of the neural compartment using the sigmoidal functions of the neurotransmitter concentrations introduced in \eqref{sigm_fdb}. We need to consider separately each type of synapse in the NMM, and the variables $x_0$, $x_1$ an $x_2$ for the feedbacks building. The NMM embeds five types of synaptic connections between populations:
\begin{itemize}
 \item $S_1$ from P' to P,
 \item $S_2$ from P to P',
 \item $S_3$ from P to I,
 \item $S_4$ from I to P,
 \item $S_5$ from P to itself
\end{itemize}
\noindent In the following we detail the modulation of neural intermediary variables for each kind of synapse separately, then we gather these changes to specify the coupling terms reproducing the glial feedback.
\begin{figure}[ht]
 \centering
 \includegraphics[width=0.7\textwidth]{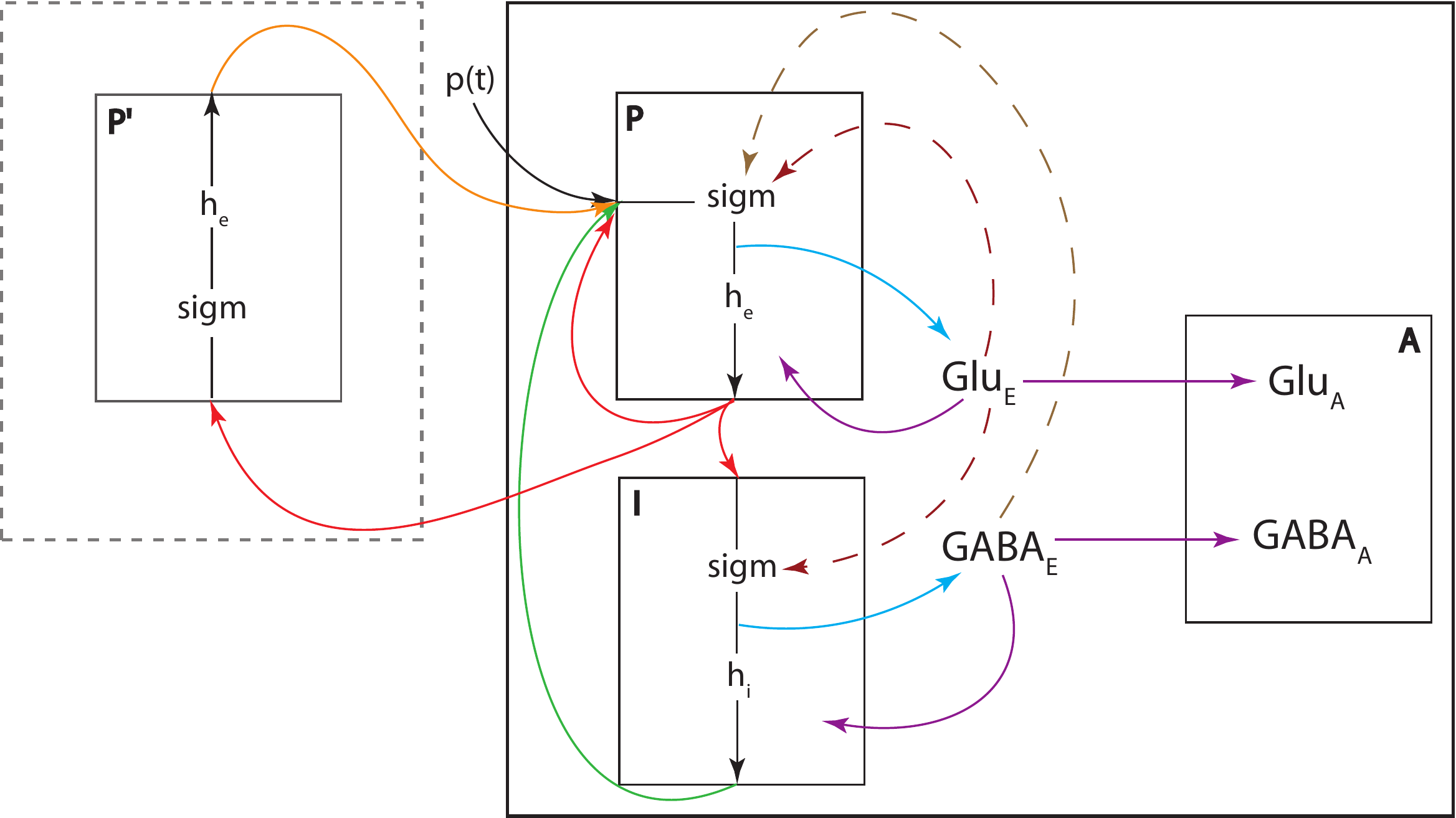}
 \caption{Neuro-glial model with glial feedback.  P and P': main and secondary populations of pyramidal cells. I: interneuron population. $p(t)$: excitatory input on population P. ${\rm [Glu]_{E}}$ and ${\rm [GABA]_{E}}$: glutamate and GABA extracellular concentrations. ${\rm [Glu]_{A}}$ and ${\rm [GABA]_{A}}$: glutamate and GABA glial concentrations. Red arrows: P$\rightarrow$P, P$\rightarrow$I and P$\rightarrow$P' couplings. Orange arrow: P'$\rightarrow$P coupling. Green arrow: I$\rightarrow$P coupling. Cyan arrows: glutamate and GABA release by populations P and I into extracellular space. Purple arrows: glial and neural reuptakes of neurotransmitters. Red dashed arrows: glutamate feedbacks on populations P and I. Brown dashed arrow: GABA feedback on population P.}
\end{figure}

In the framework of the local neuro-glial mass model, the glial feedback does not impact the synaptic connections of type $S_1$ or $S_2$. As a matter of fact, the glial compartment only takes into account neurotransmitters released locally by neurons of populations P and I, whereas population P' is non local to population P. Hence, extracellular concentrations of neurotransmitters in the vicinity of P' have no impact on the neuronal activity of P and the concentrations in the neighborhood of P and I do not influence postsynaptic neurons of population P'. In the discussion, we evoke the fact that network models based on the local model presented in this article may naturally take into account such modulation of mid-range synaptic connections. In the current study focusing on the local model, we consider the case of constant $v_{\rm P'} = v_0$.

A synaptic connection of type $S_3$ concerns the variable $x_2$. In case of glutamate excess in the extracellular space, the postsynaptic neuron is more excitable. Consequently, more neurons are activated in the population I. We model this mechanism by introducing a dependency of population $I$ excitability threshold $v_I$ on the extracellular glutamate concentration and set in equation \eqref{x2_init} :
\[
v_{\rm I} = v_0-{\rm Si_{Glui}}({\rm [Glu]_E}).
\]

On one hand, a synaptic connection of type $S_4$ is concerned by extracellular concentrations of GABA since it involves GABAergic interneurons. In case of a GABA excess in the extracellular space, the inhibition of the postsynaptic neuron is strengthened, {\it i.e.} less neurons are activated in population P which is translated in the NMM by an increase of the threshold of the sigmoidal term in the $x_0$ dynamics. On the other hand, a synaptic connection of type $S_5$ is impacted by the extracellular concentration of glutamate implying a modulation of variable $x_0$ dynamics as well. In case of an excess of glutamate in this kind of synapse, the postsynaptic neuron is more excitable. Hence, more neurons are activated in population P which can be reproduced by a decrease in the threshold parameter appearing in \eqref{x0_init}. Gathering both modulations impacting the excitability of population P, we set in equation \eqref{x0_init}
\[
 v_{\rm P} = v_0+{\rm Si_{GABA}}({\rm [GABA]_E})-{\rm Si_{Glup}}({\rm [Glu]_E}).
\]

The new neuro-glial mass model embedding the glial feedback is obtained from model \eqref{Sys_NG_SR} by considering the dynamical entries $v_{\rm I}$ and $v_{\rm P}$ mentioned above. Accordingly, the sigmoidal functions appearing in equations \eqref{x0}, \eqref{x1}, \eqref{x2}, \eqref{glune} and \eqref{gabaie} become:
\begin{align*}
 {\rm sigm}(y_1-y_2,v_{\rm P}) &= \frac{2\,e_0}{1+e^{r\,(v_0+{\rm Si_{GABA}}({\rm [GABA]_E})-{\rm Si_{Glup}}({\rm [Glu]_E})-(y_1-y_2))}} \\
 {\rm sigm}(C_1\,y_0,v_{\rm P'}) &= \frac{2\,e_0}{1+e^{r\,(v_0-C_1\,y_0)}} \\
 {\rm sigm}(C_3\,y_0,v_{\rm I}) &= \frac{2\,e_0}{1+e^{r\,(v_0-{\rm Si_{Glui}}({\rm [Glu]_E})-C_3\,y_0)}}
\end{align*}

\subsection{Effect of a GABA bolus : an illustration of the glial feedback impact}

We compare time series generated by the neuro-glial model with and without glial feedback to illustrate its impact on the model behavior (Figure \ref{Bolus_GABA}). To this aim, we mimic the same GABA bolus injection ($20$ AU) in the extracellular space at $t=30s$ with both models. For obtaining regular patterns and ease the comparison between the outputs, we consider a constant input $p(t)=p$ with a value of $p$ close to and greater than $p_{\rm SNIC}$ (Figures \ref{Bolus_GABA}, panels (a1) and (b1)). Hence, from initial time to $t=30s$, both model generate pacemaker NIS oscillations at a low frequency. Note that, even if GABA concentration is low and the corresponding sigmoidal feedback is consequently very weak, it already impacts the neural activity, which implies a difference in the spike frequency between the two LFP time series.
\begin{figure}[ht]
 \centering
 \includegraphics[width=\textwidth]{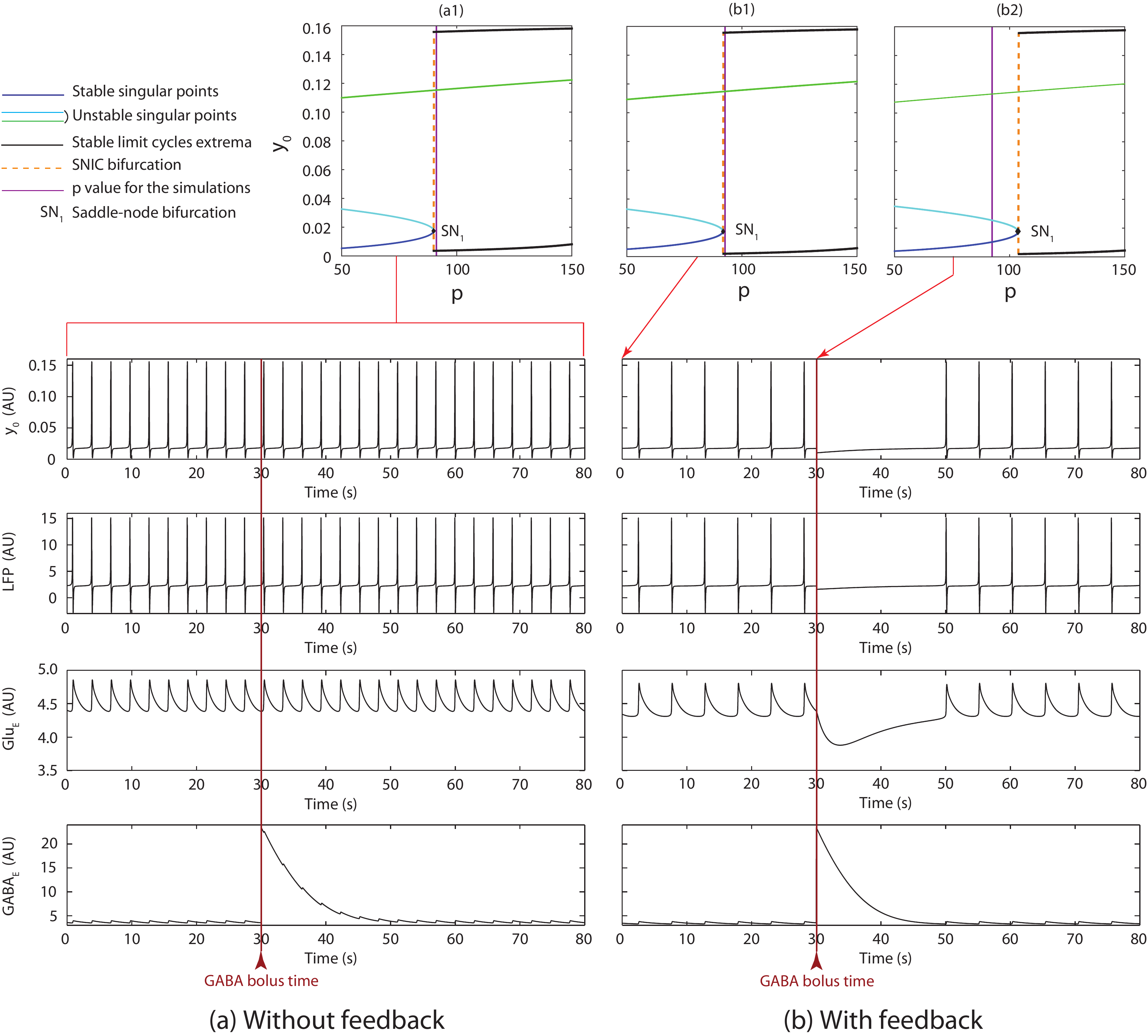}
 \caption{Bifurcation diagrams according to $p$ (top panels) computed for the model without feedback for all $t$ (a1) and for the model with feedback at $t=0s$ (b1) and $t=30s$ (b2). Time series (bottom panels) corresponding to $y_0$, LFP, extracellular concentrations of Glutamate and GABA generated by the models without feedback (left) and with feedback (right).
The purple lines on the bifurcation diagrams materialize the fixed value of input $p$. The dark red lines on the time series materialize the time of GABA bolus injection
}
 \label{Bolus_GABA}
\end{figure}

In the time series generated by the model without feedback (left panels of Figure \ref{Bolus_GABA}), the neural activity and the glutamate concentration dynamics remain unchanged (Figure \ref{Bolus_GABA}(a)) after the artificial and instantaneous increase in ${\rm [GABA]_{E}}$ that aims to mimic a GABA bolus injection. In contrast, in the model with glial feedback (right panels of Figure \ref{Bolus_GABA}), the strong increase in GABA concentration implies a break in neural activity, and thus a decrease in glutamate concentration. Once GABA concentration has become sufficiently low, neural activity starts again. The glutamate and GABA concentrations come back to their respective basal lines and oscillate under the effect of the neural spikes.

These phenomenons can be explained using bifurcation-based arguments (top panels of Figure \ref{Bolus_GABA}). In the system without feedback, the neural dynamics is entirely decoupled from the glial one. Hence, the bifurcation diagram of the NMM remains unchanged during all the simulation (Figure \ref{Bolus_GABA}(a1)).
On the other hand, in the system with feedback, the bifurcation diagram of the neural system is deformed along time, in particular the value of $p_{\rm SNIC}$ changes with the glial variables. Consequently, right after the GABA bolus, $p_{\rm SNIC}$ becomes greater than the input value $p$ (Figure \ref{Bolus_GABA}(b2)).
Moreover, we recall that $p_{\rm SNIC}$ plays the role of an activation threshold. Thus, as long as GABA concentration remains high, the neural variables are at steady state and the neural compartment remains quiescent. A direct calculation shows that the glial compartment has a single stable singular point : extracellular GABA concentration decreases towards this attractive state implying a slow decrease in $p_{\rm SNIC}$. Once, GABA concentration is low enough, $p_{\rm SNIC}$ becomes smaller than $p$, and the system oscillates again.

This analysis shows how the model with feedback can take into account changes in glutamate and GABA dynamics to modify all the dynamics of the system and illustrates the interest of embedding the glial feedback in such neuro-glial model. Our model allows us to study the effects of variations in glutamate or GABA dynamics on neural activity. In the following, we study the effects of deficiencies in the reuptake of neurotransmitters by the astrocytes both on extracellular concentrations and neural activity. For both types of deficiency (GABA and glutamate), we first describe the biologic context and mechanisms and their outcomes, then we provide a mathematical analysis of the underlying dynamical mechanisms to explain the effects that can be expected in the biological system.

\section{GABA glial deficiency} \label{GABA_deficiency}

We consider an astrocyte presenting a deficiency in its GABA transporters which implies a low capacity to reuptake the extracellular GABA. A glial cell is linked with several neurons, more specifically with several synapses. Thus the GABA in synaptic clefts linked with the defective astrocyte increases. Consequently, the concerned post-synaptic neurons receive more inhibition from the extracellular GABA and release less neurotransmitters in the following synapses. Thus, considering several defective glial cells, all neurons in a local neighborhood are affected. In summary, when glial cells present a deficiency in GABA reuptake, local neurons are more inhibited, and we expect a decrease of their activities in the corresponding simulations.

In the model, parameter $V_{\rm gba}^{\rm EA}$ stands for the maximum velocity of the GABA flux from the extracellular space to the astrocytes, {\it i.e.} in case of GABA saturation in the extracellular space. In that sense, it is related to the efficiency of the main glial transporter of GABA and modulates the glial reuptake dynamics. Consequently, to simulate a deficiency in the GABA glial reuptake, we decrease the value of this parameter. At the neuronal level we are interested in $p_{\rm SNIC}$ value according to the feedback sigmoidal functions. For sake of simplicity in this mathematical analysis, we set:
\begin{align*}
 {\rm Si_{Glui}}({\rm [Glu]_E}) & \rightarrow v_1 \\
 {\rm Si_{Glup}}({\rm [Glu]_E}) & \rightarrow \frac{m_{\rm Glup}}{m_{\rm Glui}}\,v_1 \\
 {\rm Si_{GABA}}({\rm [GABA]_E}) & \rightarrow v_2 
\end{align*}
The ranges of $v_1$ and $v_2$ are defined by the limits of ${\rm Si_{Glui}}({\rm [Glu]_E})$ and ${\rm Si_{GABA}}({\rm [GABA]_E})$ respectively:
\[
v_1 \in [0,m_{\rm Glui}] \, \text{and} \, v_2 \in [0,m_{\rm GABA}]
\]
With these new notations, the dynamical excitability thresholds $v_{\rm P}$, $v_{\rm P'}$ and $v_{\rm I}$ of populations P, P' and I become:
\begin{align*}
 v_{\rm P} &= v_0+v_2-\frac{m_{\rm Glup}}{m_{\rm Glui}}\,v_1 \\
 v_{\rm P'} &= v_0 \\
 v_{\rm I} &= v_0-v_1
\end{align*}

With these new parameters, an increase or a decrease in GABA (resp. glutamate) extracellular concentration is represented by an increase or a decrease in the value of $v_2$ (resp. $v_1$) respectively. The natural effect of a deficiency of GABA glial reuptake on neural activity is an increase in the extracellular GABA concentration. Thereby, we characterize the dependency of $p_{\rm SNIC}$ on the value of $v_2$. The assumption that $v_1$ can be kept constant is justified in the following Remark \ref{v1constant}.

\begin{proposition} \label{Prop1}
 $p_{\rm SNIC}$ is linear and increasing according to $v_2$
\end{proposition}
\begin{proof}
The set of the system singular points obtained for the different values of parameter $p$ can be explicitly expressed according to $y_0$, $v_1$ and $v_2$ all other parameters being fixed. The $y_0$ components of the singular points for given values of $p$, $v_1$ and $v_2$ are characterized as solutions of
\begin{equation}
 p = f(y_0,v_1,v_2)
\end{equation}
where
\begin{align} \label{f_y0_v1_v2}
 f(y_0,v_1,v_2) &= \frac{a}{A}\,(v_0-\frac{m_{\rm Glup}}{m_{\rm Glui}}\,v_1+v_2)-\frac{a}{A\,r}{\rm ln}\left( \frac{2\,A\,e_0-a\,y_0}{a\,y_0} \right)-C_2\,{\rm sigm}(C_1\,y_0,v_0) \nonumber \\
 & \qquad +\frac{a\,B}{b\,A}\,C_4\,{\rm sigm}(C_3\,y_0,v_0-v_1)-\frac{a\,G}{A}\,y_0.
\end{align}
All the other components of a given singular point result by direct calculation from its $y_0$ component. We rewrite equation \eqref{f_y0_v1_v2} as follows
\begin{equation} \label{f_lin_v2}
 f(y_0,v_1,v_2)=\frac{a}{A}\,v_2+q(y_0,v_1)
\end{equation}
Obviously the two saddle-node bifurcation values $p_{\rm SN_1}$ and $p_{\rm SN_2}$ are local extrema of function $f(y_0,v_1,v_2)$. In particular $p_{\rm SN_1}=p_{\rm SNIC}$ is the local maximum of $f(y_0,v_1,v_2)$ and is defined as the solution of
\begin{subequations}
\label{ySNIC}
\begin{align}
 & p = f(y_0,v_1,v_2) \label{pSNIC_ySNIC} \\
 & \frac{\partial f}{\partial y_0}(y_0,v_1,v_2) = 0 \\
 & \frac{\partial^2 f}{\partial y_0^2}(y_0,v_1,v_2) \leqslant 0
\end{align}
\end{subequations}
Since $\frac{\partial f}{\partial y_0}(y_0,v_1,v_2)$ is independent on $v_2$, so is $y_{\rm SNIC}$ and it can be considered as a parameter in equation \eqref{pSNIC_ySNIC}. From \eqref{f_lin_v2} and \eqref{ySNIC} we obtain the following expression for $p_{\rm SNIC}$
\begin{equation*}
 p_{\rm SNIC} = \frac{a}{A}\,v_2+q(y_{\rm SNIC},v_1)
 \label{pSNIC_v2}
\end{equation*}
\end{proof}

Let us consider the model generating an oscillatory output with a fixed value of $p$ (($p_{\rm SNIC}<p$). If the extracellular concentration of GABA increases (e.g. by an injection of a GABA bolus as in Figure \ref{Bolus_GABA}), the value of $v_2$ increases and Proposition 1 asserts that the value of $p_{\rm SNIC}$ also increases. As already explained, the closest $p_{\rm SNIC}$ is to $p$, with $p_{\rm SNIC}<p$, the largest is the limit cycle period, thus the oscillation frequency of the outputs decreases. If $p_{\rm SNIC}$ increases enough such that $p_{\rm SNIC}>p$, the stable limit cycle of the system disappears, and the neural compartment becomes quiescent.

In the case of a deficiency of GABA glial reuptake, the extracellular concentration of GABA increases, and we can use Proposition 1 to explain the subsequent effects. For that, we use the following {\it in silico} protocol: we initialize the neuro-glial model in an oscillatory phase with a low oscillation frequency and consider $p(t)$ a Gaussian input. At $t=40s$, we turn off the GABA glial reuptake by setting $V_{\rm gba}^{\rm EA}=0$ (Figure \ref{v2_Psnic}). The result is an increase in GABA extracellular concentration implying an increase in $p_{\rm SNIC}$. As $p_{\rm SNIC}$ increases, the probability for $p(t)$ to overcome $p_{\rm SNIC}$ along the associated brownian motion decreases, and also does the oscillation frequency (Figure \ref{v2_Psnic}). Consequently, we observe a decrease in the oscillation frequency after $t>40s$. In the time series, the oscillation frequency decreases gradually during a transient ($40s<t<60s$) until reaching its minimum. This can be explained by the slow increase of GABA extracellular concentration that reaches its new baseline at $t=60s$.

\begin{remark} \label{v1constant}
A deficiency in the GABA reuptake by the astrocyte implies a decrease in the neural activity. Hence, the glutamate extracellular concentration remains close to the baseline. Consequently, the impact of the changes in $v_1$ value can be neglected and, under this approximation, Proposition \ref{Prop1} characterizes the global effect of such deficiency on the neural compartment excitability.
\end{remark}

\begin{figure}[htbp]
 \centering
 \includegraphics[width=\textwidth]{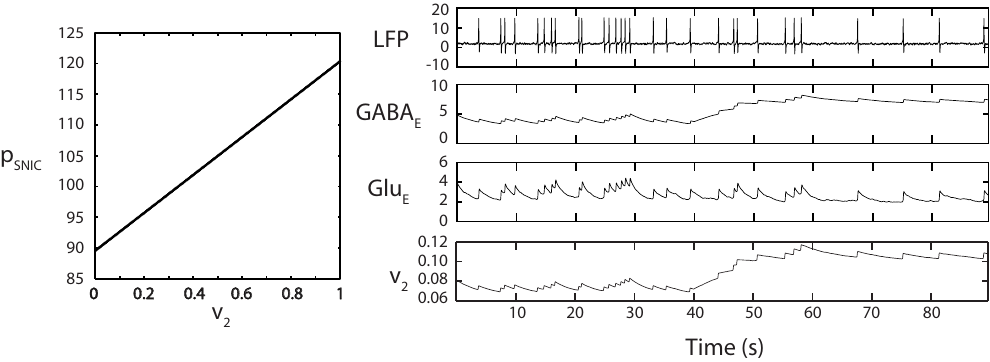}
\caption{Variation of $p_{\rm SNIC}$ value according to $v_2$ (left). Time series corresponding to ${\rm LFP}$, ${\rm [GABA]_E}$, ${\rm [Glu]_E}$ and $v_2={\rm Si_{GABA}}({\rm [GABA]_E})$ (right from top to bottom) for $p(t)$ a Gaussian variable. At $t=40$ s, the GABA glial reuptake is artificially altered by setting $V_{\rm gba}^{\rm EA}=0$.}
 \label{v2_Psnic}
\end{figure}

\section{Glutamate glial deficiency} \label{Glutamate_deficiency}

In this section, we investigate the impact of a deficiency of glutamate reuptake by the astrocytes upon the neuronal activity.
Such deficiency provoke an increase in the extracellular concentration of glutamate and, consequently, neurons in this neighborhood are more excitable. Yet, it is important to note that interneurons release more GABA implying an increase in the GABA extracellular concentration as well, and an enhancement of the inhibition of the pyramidal activity. Hence, the possible balance between glutamate-induced over-excitation and subsequent GABA-induced over-inhibition may lead to different types of response of the neuronal compartment. As previously, for studying theoretically the underlying mechanisms, we consider the NMM with two parameters $v_1$ and $v_2$ representing glutamate- and GABA-related feedbacks respectively. In the last section, we have characterized the linearity of $p_{\rm SNIC}$ according to $v_2$ for any fixed value of $v_1$. Now, let us fix the value of $v_2$ and study the variations of $p_{\rm SNIC}$ according to $v_1$.

We recall that $p_{\rm SNIC}$ is a key value of the system structure since it represents the excitability threshold of the neural compartment. It is important to note that in a specific case, the SNIC bifurcation disappears without disappearance of this excitability threshold. In this case, the supercritical Hopf bifurcation occurring for a large value of $p$ and the saddle-node bifurcation (SN$_1$) previously linked to the SNIC bifurcation, are preserved and a subcritical Hopf bifurcation appears, close to SN$_1$, giving birth to an unstable limit cycle. This limit cycle persists for a very small interval of $p$ values before it disappears through a fold bifurcation of limit cycles. We refer the reader to \cite{Garnier_2015} for more details about this bifurcation structure.
For certain parameter values of the whole model, the bifurcation structure may therefore be lost when $v_1$ varies. Computing the region of the parameter space (of high dimension) for which it remains unchanged for any $v_1$ is difficult. Yet, the previous analysis of the NMM \cite{Garnier_2015} ensures us that this region is large. Thus, in the following, we assume that $p_{\rm SNIC}$ exists and the associated saddle-node bifurcation is not degenerated for all $v_1 \in [0, m_{\rm Glui}]$, {\it i.e.} the maximal interval of values taken by $v_1 = {\rm Si_{Glui}}({\rm [Glu]_E})$, which is the case, in particular, for the parameter values given in Table \ref{par_table} and \ref{par_feedback}, that have been used for the simulations.

We recall that $p_{\rm SNIC}$ can be written as follows:
\begin{equation} \label{pSNIC_f}
 p_{\rm SNIC} = f(y_{\rm SNIC},v_1,v_2)
\end{equation}
where $f$ is given by \eqref{f_y0_v1_v2}.
Since we consider $v_2$ fixed we introduce the function
\[
g(y_0,v_1) \equiv f(y_0,v_1,v_2)|_{v_2\, \text{fixed}}
\]
As explained above, for each $v_1$, there exists a unique bifurcation value $p_{\rm SNIC}$ occurring at the non-hyperbolic (saddle-node) singular point characterized by $y_{\rm SNIC}$ which is defined by
\[
\begin{cases}
&\frac{\partial g}{\partial y_0}(y_{\rm SNIC},v_1) = 0 ,\\
& \frac{\partial^2 g}{\partial y_0^2}(y_{\rm SNIC},v_1) < 0.
 \end{cases}
\]
This value satisfies $p_{\rm SNIC} = g(y_{\rm SNIC},v_1)$. We can not find the explicit expressions of $y_{\rm SNIC}(v_1)$ and $p_{\rm SNIC}(v_1)$. Thus, for characterizing the variations of $p_{\rm SNIC}$ with $v_1$, we take advantage of the implicit definitions above and focus on localizing the extrema of $p_{\rm SNIC}(v_1)$.

\begin{proposition}
Assume that for all $v_1 \in [0,m_{\rm Glui}]$, $p_{\rm SNIC}$ exists and the associated saddle-node bifurcation is not degenerate. Then
 \begin{enumerate}
  \item if $\frac{m_{\rm Glup}}{m_{\rm Glui}} \geqslant \frac{B\,e_0\,r\,C_4}{2\,b}$, $p_{\rm SNIC}(v_1)$ has no local extremum, \vspace{0.2cm}
  \item if $\frac{m_{\rm Glup}}{m_{\rm Glui}} \in \left]0, \frac{B\,e_0\,r\,C_4}{2\,b}\right[$, $p_{\rm SNIC}(v_1)$ may admit two local extrema: a minimum at $v_1^*$ and a maximum at $v_1^{**}$. If both exist, then $v_1^* < v_1^{**}$. 
 \end{enumerate}
\end{proposition}
\begin{proof}
Let us search for local extrema of function $p_{\rm SNIC}(v_1)$ which is implicitly defined by \eqref{ySNIC}. Hence, we are interested in solving the following problem of optimization under constraint :
\begin{equation}
\label{pSNIC_v1}
\min / \max \left\{g(y_0,v_1) \quad \lvert \quad \frac{\partial g}{\partial y_0}(y_0,v_1) = 0 \right\}
\end{equation}
We introduce the associated lagrangian function
\begin{equation*}
 {\sf L}(y_0,v_1,\lambda) = g(y_0,v_1)-\lambda\,\frac{\partial g}{\partial y_0}(y_0,v_1)
\end{equation*}
The necessary condition for the existence of an extremum of $g$ under the constraint $\frac{\partial g}{\partial y_0}=0$ is 
\begin{equation*}
 \overrightarrow{\nabla}{\sf L}(y_0,v_1,\lambda)=0
\end{equation*}
that is
\begin{subequations}
\label{CN_lagrangien}
\begin{align}
 & \frac{\partial g}{\partial y_0}(y_0,v_1)-\lambda\,\frac{\partial^2 g}{\partial y_0^2}(y_0,v_1) =0, \\
 & \frac{\partial g}{\partial v_1}(y_0,v_1)-\lambda\,\frac{\partial^2 g}{\partial v_1  \partial y_0}(y_0,v_1) =0, \\
 & \frac{\partial g}{\partial y_0}(y_0,v_1) = 0.
\end{align}
\end{subequations}
By assumption, the saddle-node bifurcation associated with the SNIC bifurcation is not degenerate, {\it i.e.} every solution of \eqref{pSNIC_v1} for $v_1 \in [0, m_{\rm Glui}]$ satisfies $\frac{\partial^2 g}{\partial y_0^2}(y_0,v_1) \neq 0$. Thus, system \eqref{CN_lagrangien} reads
\begin{subequations} \label{SPLag}
 \begin{align}
  & \lambda = 0, \\
  & \frac{\partial g}{\partial v_1}(y_0,v_1) = 0, \\
  & \frac{\partial g}{\partial y_0}(y_0,v_1) = 0. \label{dgdy0}
 \end{align}
\end{subequations}
Consequently, if the problem under constraint admits an extremum, this extremum satisfies $\frac{\partial g}{\partial v_1}=0$. Following the assumption that a SNIC bifurcation occurs for any value of $v_1 \in [0, m_{\rm Glui}]$, equation \eqref{dgdy0} admits a solution for any $v_1$. Hence, if the problem under constraint admits an extremum, it corresponds to a SNIC bifurcation occurring at $(y_0,v_1)$ such that
\[
 \frac{\partial g}{\partial y_0}(y_0,v_1) = 0.
\]
From \eqref{f_y0_v1_v2}, we obtain
\begin{equation*}
 \frac{\partial g}{\partial v_1}(y_0,v_1) = - \frac{a}{A}\left( \frac{m_{\rm Glup}}{m_{\rm Glui}}+\frac{B}{b}\,C_4\,\frac{\partial {\rm sigm}}{\partial v}(C_3\,y_0,v_0-v_1) \right).
\end{equation*}
Using the facts that, for any fixed values of $y_0$, function $v_1 \to \frac{\partial g}{\partial v_1}(y_0,v_1)$ is bell-shaped and its maximal value does not depend on $y_0$ (see Figure \ref{diff_sigm_v1}), one obtains that function $\frac{\partial g}{\partial v_1}(y_0,v_1)$ vanishes in $v_1$ if
\begin{equation} \label{Cond_mglu}
 \frac{m_{\rm Glup}}{m_{\rm Glui}} \in \left]0, \frac{B\,e_0\,r\,C_4}{2\,b}\right[.
\end{equation}
If $\frac{m_{\rm Glup}}{m_{\rm Glui}} \geqslant \frac{B\,e_0\,r\,C_4}{2\,b}$, function $\frac{\partial g}{\partial v_1}(y_0,v_1)$ admits no zero, which proves the first item of the Proposition 2.

Now, we assume that condition \eqref{Cond_mglu} is fulfilled and we search the values of $v_1$ satisfying $\frac{\partial g}{\partial v_1}(y_0,v_1)=0$, {\it i.e.}
\begin{equation*}
 \frac{m_{\rm Glup}}{m_{\rm Glui}}+\frac{B}{b}\,C_4\,\frac{\partial {\rm sigm}}{\partial v}(C_3\,y_0,v_0-v_1)=0.
\end{equation*}
which reads
\begin{equation} \label{diff_sigm}
 \left(e^{r\,\left( v_0-v_1-C_3\,y_0 \right)}\right)^2\,\frac{m_{\rm Glup}}{m_{\rm Glui}}+e^{r\,\left( v_0-v_1-C_3\,y_0 \right)}\,\left( 2\,\frac{m_{\rm Glup}}{m_{\rm Glui}}-2\,\frac{B}{b}e_0\,r\,C_4 \right)+\frac{m_{\rm Glup}}{m_{\rm Glui}}=0.
\end{equation}
Setting
\begin{equation} \label{V+-}
 V_\pm = \frac{B\,e_0\,r\,C_4-b\,\frac{m_{\rm Glup}}{m_{\rm Glui}}\pm \sqrt{(B\,e_0\,r\,C_4)^2-2\,B\,e_0\,r\,C_4\,\frac{m_{\rm Glup}}{m_{\rm Glui}}}}{b\,\frac{m_{\rm Glup}}{m_{\rm Glui}}}
\end{equation}
we obtain the two solutions $v^* < v^{**}$ of $\frac{\partial g}{\partial v_1}(y_0,v_1)=0$ :
\begin{eqnarray} \label{v1+-}
 v_1^* = v_0-C_3\,y_0-\frac{1}{r}\,{\rm ln}\left( V_+ \right) \\
 v_1^{**} = v_0-C_3\,y_0-\frac{1}{r}\,{\rm ln}\left( V_- \right) 
\end{eqnarray}
Note that $v_1^*$ (resp. $v_1^{**}$) corresponds to the extremum when the saddle-node SN$_1$ (resp. SN$_2$) crosses the fold of the surface $g(y_0,v_1)=p$. We consider $v_1=v_1^*$ and we note $y_0^*$ the value of $y_0$ corresponding to the SNIC connection for this value of $v_1$, {\it i.e.} the solution of 
\begin{align*}
 \frac{\partial g}{\partial y_0}(y_0,v_1^*) &= 0 \\
 \frac{\partial^2 g}{\partial y_0^2}(y_0,v_1^*) & < 0
\end{align*}
To prove that $p_{\rm SNIC}$ reaches a local minimum at $v_1=v_1^*$, we introduce the bordered Hessian matrix $\overline{H}$ associated with the lagrangian function at its singular point $(y_0,v_1,\lambda) = (y_0^*,v_1^*,0)$ (solution of system \eqref{SPLag}):
\[
 \overline{H}(y_0^*, v_1^*,0)=
 \begin{pmatrix}
  0 & \frac{\partial^2 g}{\partial y_0^2} & \frac{\partial^2 g}{\partial v_1 \partial y_0} \\
  \frac{\partial^2 g}{\partial y_0^2} & \frac{\partial^2 {\sf L}}{\partial y_0^2} & \frac{\partial^2 {\sf L}}{\partial v_1 \partial y_0} \\
  \frac{\partial^2 g}{\partial v_1 \partial y_0} & \frac{\partial^2 {\sf L}}{\partial v_1 \partial y_0} & \frac{\partial^2 {\sf L}}{\partial v_1^2}
 \end{pmatrix}_{|_{(y_0^*, v_1^*,0)}}=
 \begin{pmatrix}
  0 & \frac{\partial^2 g}{\partial y_0^2} & \frac{\partial^2 g}{\partial v_1 \partial y_0} \\
  \frac{\partial^2 g}{\partial y_0^2} & \frac{\partial^2 g}{\partial y_0^2} & \frac{\partial^2 g}{\partial v_1 \partial y_0} \\
  \frac{\partial^2 g}{\partial v_1 \partial y_0} & \frac{\partial^2 g}{\partial v_1 \partial y_0} & \frac{\partial^2 g}{\partial v_1^2}
 \end{pmatrix}_{|_{(y_0^*, v_1^*,0)}}
\]
The determinant of $\overline{H}(y_0^*,v_1^*,0)$ is given by
\[
\det{\overline{H}(y_0^*,v_1^*,0)}= -\frac{\partial^2 g}{\partial y_0^2}(y_0^*, v_1^*)\,\left[ \frac{\partial^2 g}{\partial y_0^2}(y_0^*, v_1^*)\,\frac{\partial^2 g}{\partial v_1^2}(y_0^*, v_1^*)-\left(\frac{\partial^2 g}{\partial v_1 \partial y_0}(y_0^*, v_1^*)\right)^2 \right]
\]
On one hand, the saddle-node associated with the SNIC bifurcation is not degenerate and is a local maximum of $g(y_0,v_1)$, thus $\frac{\partial^2 g}{\partial y_0^2}(y_0^*, v_1^*)<0$. On the other hand, for any $y_0$, $v_1 \to \frac{\partial g}{\partial v_1}(y_0,v_1)$ is increasing at $(y_0, v_1^*)$ (see Figure \ref{diff_sigm_v1}), thus $\frac{\partial^2 g}{\partial v_1^2}(y_0^*, v_1^*)>0$. Finally
\[
\det{\overline{H}(y_0^*,v_1^*,0)} < 0
\]
and $(y_0^*,v_1^*)$ corresponds to a local minimum of $p_{\rm SNIC}$. A similar argument proves that $(y_0^{**},v_1^{**})$ corresponds to a local maximum of $p_{\rm SNIC}$ (where $y_0^{**}$ is the $y_0$ value corresponding to ${\rm SN_2}$ bifurcation for $v_1=v_1^{**}$).
\end{proof}

The above proposition can be interpreted as a necessary condition for having a change in the sense of variations of $p_{\rm SNIC}$ when $v_1$ varies in $[0, m_{\rm glui}]$. The following result gives a sufficient condition for $v_1^*$ actually lying in $[0, m_{\rm glui}]$.

\begin{corollary}
We have $v_1^*  \in [0,m_{\rm Glui}]$ if and only if $\frac{m_{\rm Glup}}{m_{\rm Glui}} \in \left[I_1,I_2\right]$ where
\begin{align}
  I_1 &= \frac{2\,B\,e_0\,r\,C_4}{b}\,\frac{e^{r\,(v_0-C_3\,y_0^*)}}{(1+e^{r\,(v_0-C_3\,y_0^*)})^2}, \label{I1} \\
  I_2 &= \frac{2\,B\,e_0\,r\,C_4}{b}\,\frac{e^{r\,(v_0-m_{\rm Glui}-C_3\,y_0^*)}}{(1+e^{r\,(v_0-m_{\rm Glui}-C_3\,y_0^*)})^2} \label{I2}
\end{align}
\end{corollary}
\begin{proof}

Since $v_1^*$ satisfies $\frac{\partial g}{\partial v_1}(y_0^*,v_1^*)=0$, one obtains, from equation \eqref{diff_sigm},
\begin{equation} \label{h_v1}
 \frac{m_{\rm Glup}}{m_{\rm Glui}} = \frac{2\,B\,e_0\,r\,C_4}{b}\,\frac{e^{r\,(v_0-v_1^*-C_3\,y_0^*)}}{(1+e^{r\,(v_0-v_1^*-C_3\,y_0^*)})^2} = h(v_1^*)
\end{equation}
For any $y_0$, function $v_1 \mapsto \frac{e^{r\,(v_0-v_1-C_3\,y_0)}}{\left(1+e^{r\,(v_0-v_1-C_3\,y_0)}\right)^2}$ is strictly increasing over $[0,m_{\rm Glui}]$ and $v_1^* \in [0,m_{\rm Glui}]$ if and only if $\frac{m_{\rm Glup}}{m_{\rm Glui}}\in[I_1,I_2]$ defined by \eqref{I1} and \eqref{I2} (see Figure \ref{diff_sigm_v1}).
\begin{figure}[htbp]
 \centering
 \includegraphics[width=0.45\textwidth]{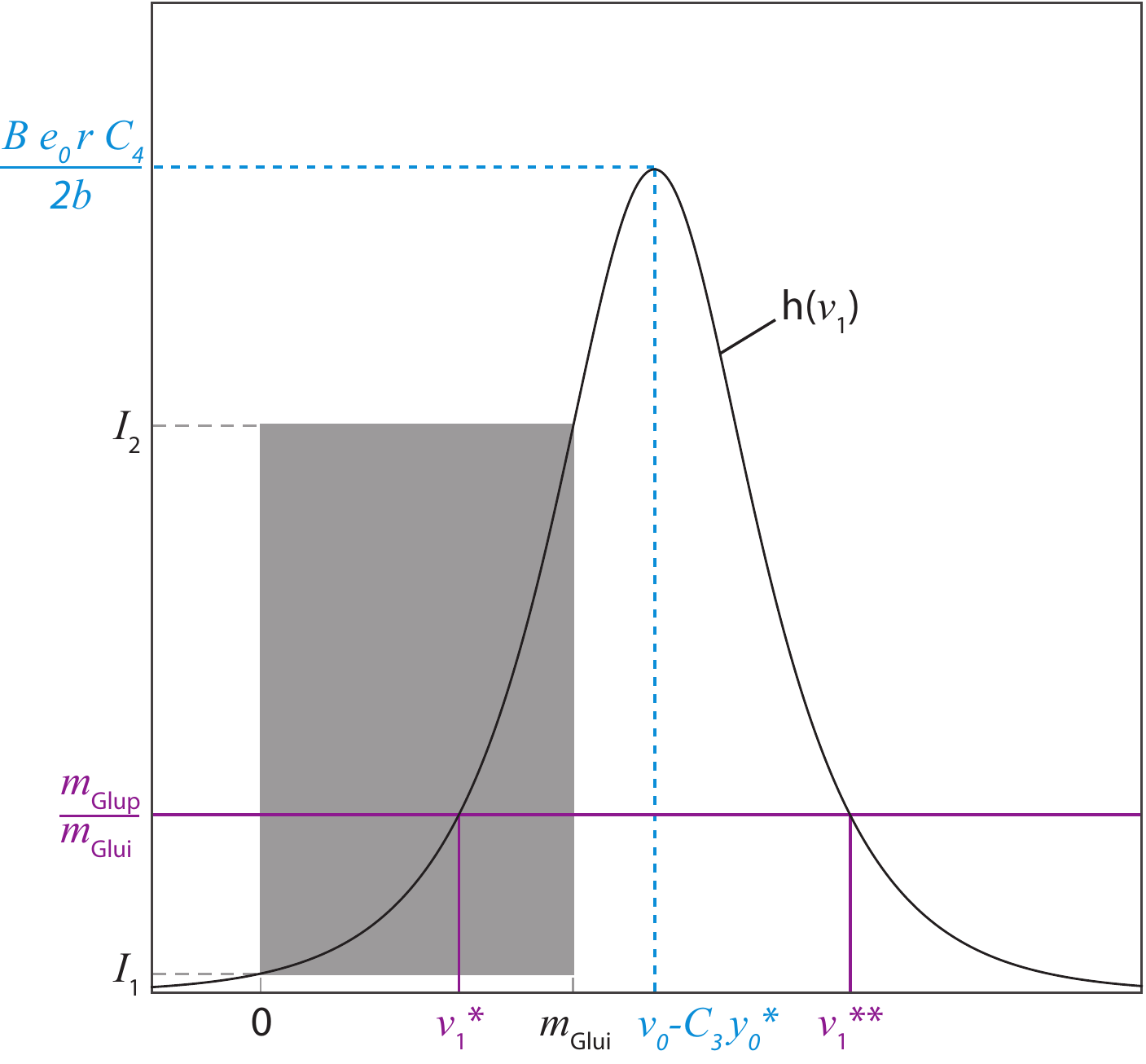}
 \caption{Graphic representation of function $h$ defined by \eqref{h_v1} and interval $[I_1,I_2]$ for which $v_1^* \in [0,m_{\rm glui}]$.}
 \label{diff_sigm_v1}
\end{figure}
\end{proof}

\noindent In conclusion, for a fixed value of $v_2$, $p_{\rm SNIC}$ reaches a local minimum at a value $v_1^* \in [0,m_{\rm Glui}]$ if and only if
\[
\frac{m_{\rm Glup}}{m_{\rm Glui}} \in \left]0, \frac{B\,e_0\,r\,C_4}{2\,b}\right[ \, \cap \, \left[I_1,I_2\right].
\]
Moreover, in section \ref{GABA_deficiency} about the GABA glial deficiency, we proved that, for a fixed value of $v_1$, $p_{\rm SNIC}$ is linear and increasing with $v_2$. Both results allow us to predict that there exist three shapes of $p_{\rm SNIC}(v_1,v_2)$ according to the value of  $\frac{m_{\rm Glup}}{m_{\rm Glui}}$.
\begin{itemize}
\item If $\frac{m_{\rm Glup}}{m_{\rm Glui}}<I_1$ then $v_1^*<0$ and $p_{\rm SNIC}$ strictly increases with $v_1$ and $v_2$. \vspace{0.2cm}
\item If $\frac{m_{\rm Glup}}{m_{\rm Glui}}>I_2$, then $v_1^*>m_{\rm Glui}$ and $p_{\rm SNIC}$ strictly decreases when $v_1$ increases (for $v_2$ fixed) and strictly increases with $v_2$ (for $v_1$ fixed). \vspace{0.2cm}
\item If $\frac{m_{\rm Glup}}{m_{\rm Glui}} \in \left[I_1,I_2\right]$, then $v_1^* \in [0,m_{\rm Glui}]$ and $p_{\rm SNIC}$ decreases when $v_1$ increases in $[0,v_1^*]$ and increases with $v_1>v_1^*$ (for $v_2$ fixed).
\end{itemize}
In the following we illustrate the three qualitative types of neural activity resulting from an astrocyte deficiency to capture glutamate using the following values :
\begin{align*}
 a) \quad \frac{m_{\rm Glup}}{m_{\rm Glui}} = 1.7 <I_1, \qquad 
 b) \quad \frac{m_{\rm Glup}}{m_{\rm Glui}} = 3.2 >I_2, \qquad 
 c) \quad \frac{m_{\rm Glup}}{m_{\rm Glui}} = 2.43 \in \left[I_1,I_2\right]
\end{align*}
For each case, we provide simulations representing the value of $p_{\rm SNIC}$ in $(v_1,v_2)$ space and time series generated by the model when the glutamate glial reuptake is altered (Figures \ref{mglup2}, \ref{mglup3} and \ref{mglup25}).

\begin{figure}[htbp]
 \centering
 \includegraphics[width=0.93\textwidth]{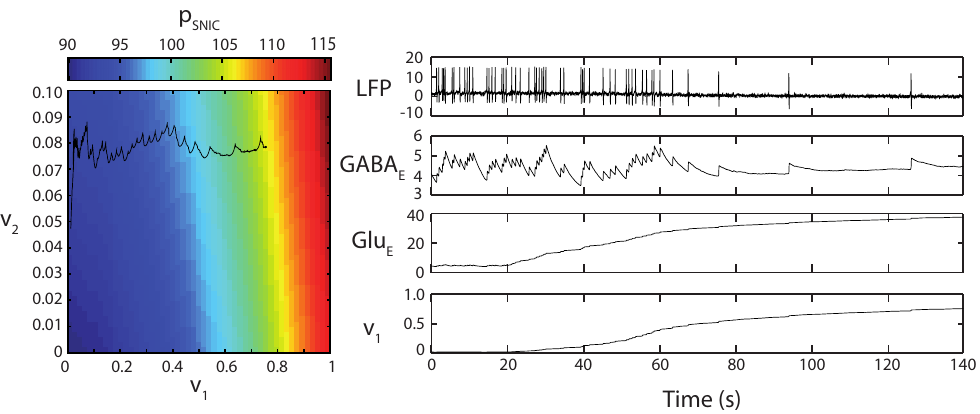}
 \caption{Colormap displaying the value of $p_{\rm SNIC}$ in $(v_1,v_2) \in [0,1] \times [0,0.1]$ space (left panel), and time series corresponding to ${\rm LFP}$, ${\rm [GABA]_E}$, ${\rm [Glu]_E}$ and $v_1={\rm Si_{Glui}}({\rm [Glu]_E})$ (right panels) obtained with $\frac{m_{\rm Glup}}{m_{\rm Glui}}=1.7$. Black curve on the colormap: trace of (${\rm Si_{Glui}}({\rm [Glu]_E})$, ${\rm Si_{GABA}}({\rm [GABA]_E})$) along the orbits. At $t=20s$, we alter the glutamate glial reuptake by setting $V_{\rm glu}^{\rm EA}=0$.}
 \label{mglup2}
\end{figure}

\begin{figure}[htbp]
 \centering
 \includegraphics[width=0.95\textwidth]{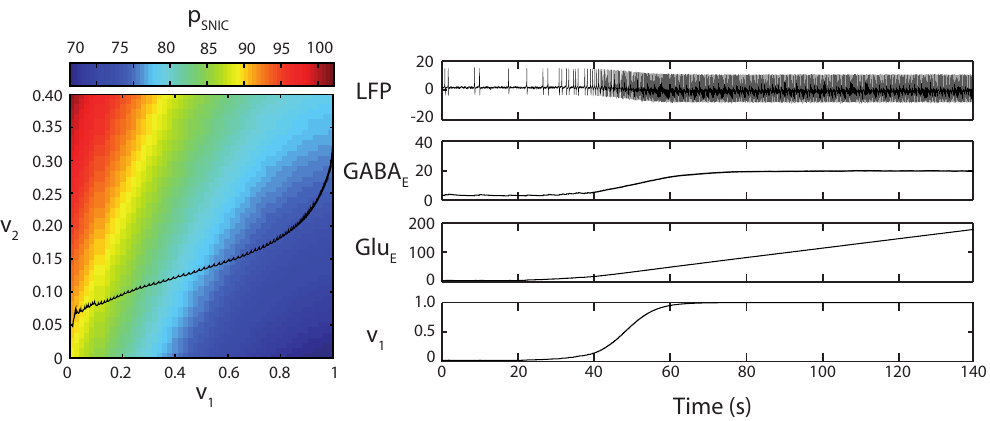}
 \caption{Colormap displaying the value of $p_{\rm SNIC}$ in $(v_1,v_2) \in [0,1] \times [0,0.4]$ space (left panel), and time series corresponding to ${\rm LFP}$, ${\rm [GABA]_E}$, ${\rm [Glu]_E}$ and $v_1={\rm Si_{Glui}}({\rm [Glu]_E})$ (right panels) obtained with $\frac{m_{\rm Glup}}{m_{\rm Glui}}=3.2$. Black curve on the colormap: trace of (${\rm Si_{Glui}}({\rm [Glu]_E})$, ${\rm Si_{GABA}}({\rm [GABA]_E})$) along the orbits. At $t=20s$, we alter the glutamate glial reuptake by setting $V_{\rm glu}^{\rm EA}=0$.}
 \label{mglup3}
\end{figure}

\begin{figure}[htbp]
 \centering
 \includegraphics[width=0.95\textwidth]{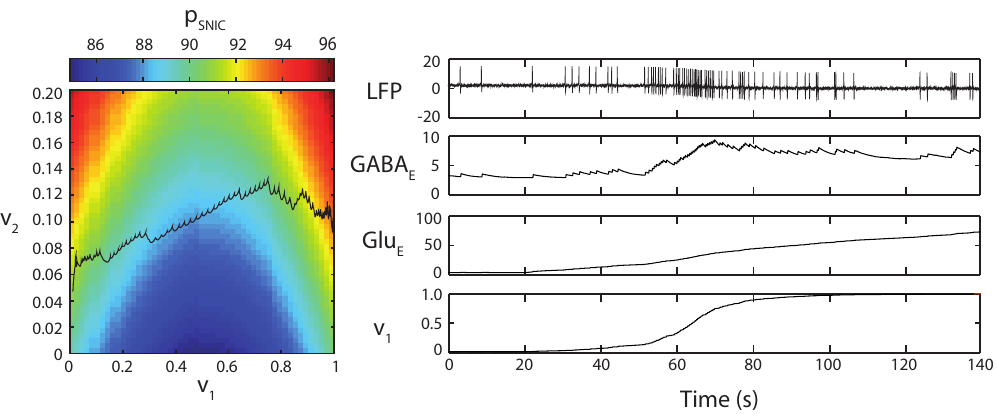}
 \caption{Colormap displaying the value of $p_{\rm SNIC}$ in $(v_1,v_2) \in [0,1] \times [0,0.2]$ space (left panel), and time series corresponding to ${\rm LFP}$, ${\rm [GABA]_E}$, ${\rm [Glu]_E}$ and $v_1={\rm Si_{Glui}}({\rm [Glu]_E})$ (right panels) obtained with $\frac{m_{\rm Glup}}{m_{\rm Glui}}=2.43$. Black curve on the colormap: trace of (${\rm Si_{Glui}}({\rm [Glu]_E})$, ${\rm Si_{GABA}}({\rm [GABA]_E})$) along the orbits. At $t=20s$, we alter the glutamate glial reuptake by setting $V_{\rm glu}^{\rm EA}=0$.}
 \label{mglup25}
\end{figure}

For $\frac{m_{\rm Glup}}{m_{\rm Glui}}=1.7<I_1$ (case $a)$, figure \ref{mglup2}), $v_1^*$ is negative and thus $p_{\rm SNIC}$ increases with $v_1$. Thereby, when we reproduce a glutamate glial deficiency triggering an increase in $v_1$ value, we observe a decrease in the oscillation frequencies in the neural activity. Moreover, reduction of the glutamate glial reuptake together with the strong decrease of neural activity, triggers an increase in the baseline of glutamate extracellular concentration. Indeed, the lack of reuptake involves the accumulation of extracellular glutamate. Since neurons are less activated, the glutamate release is lowered and the neural reuptake can stabilize the glutamate extracellular concentration.

Figure \ref{mglup3} illustrates case $b)$ ($\frac{m_{\rm Glup}}{m_{\rm Glui}}=3.2>I_2$). Since $v_1^*>1$, $p_{\rm SNIC}$ decreases as $v_1$ increases. Thereby, a glutamate glial deficiency triggers an increase of $v_1$ value, and we observe an increase in the oscillation frequencies in the neural time-series. Since the glial reuptake is reduced, glutamate accumulates in the extracellular space and the corresponding concentration baseline increases drastically.

Figure \ref{mglup25} illustrates intermediary case $c)$ ($\frac{m_{\rm Glup}}{m_{\rm Glui}}=2.43\in [I_1,I_2]$). Since $v_1^* \in ]0,1[$, $p_{\rm SNIC}$ decreases for $v_1<v_1^*$ and increases otherwise. Thereby, when we simulate a glutamate glial deficiency triggering an increase in $v_1$ value, we observe an increase in the oscillation frequencies in LFP time series, followed by a decrease. Indeed, after the alteration of the glutamate glial reuptake, the glutamate extracellular concentration increases and excites more pyramidal cells and interneurons. Subsequently, interneurons release more GABA, which implies an increase of the inhibition of pyramidal cells and a decrease in the neural activity. This sequence of events explains the delay in the regulation of the oscillation frequencies.

This kind of behavior is physiologically relevant. Indeed, it is conceivable that an excess of glutamate extracellular concentration is regulated after a delay, triggering a decrease of neural activity after the initial increase. Moreover, the frequency after the regulation delay can be greater or lower than the initial one, depending on the value of the ratio $\frac{m_{\rm Glup}}{m_{\rm Glui}}$. Note that this value can be tuned to obtain $v_1^*$ small enough and $p_{\rm SNIC}$ large enough so that the frequency after regulation is equal or lower than the one before reuptake deficiency. This property offers the possibility of fitting the model outputs to experimental data and allows us to propose hypotheses about physiological and pathological mechanisms.

\section{Conclusion and  Discussion}

In this article, we have introduced a new neuro-glial mass model built on a bilateral coupling of the NMM studied in \cite{Garnier_2015} and the glial model proposed in \cite{Blanchard_2015} focusing on GABA and glutamate concentration dynamics. The model is based on recent biological knowledge resulting from experimental data \cite{Bellone_2008, Huang_2004, Niswender_2010, Pittenger_2011} to ground the interaction between the neural and extracellular/glial compartments. Note that, as explained in the beginning of section 2.2., only basic properties of the dynamical coupling are needed to prove the qualitative dependencies studied in sections 3 and 4. By lack of experimental data in the literature, we chose sigmoidal functions for relaying the glial feedbacks in the simulations because they represent a paragon of bounded increasing functions involving a significant threshold effect. Then, using the interpretation of the aggregated -- yet biophysical -- parameters involved in this model, we have reproduced {\it in silico} various types of deficiencies in the reuptake of GABA or Glutamate by the astrocytes and studied their impact upon the neural activity. We took advantage of the bifurcation analysis performed in \cite{Garnier_2015} to characterize theoretically the dynamical mechanism leading to local neuronal hyperexcitability through a modulation of the SNIC bifurcation standing for the excitability threshold of the neural compartment. 

The first result concerns the impact of a deficiency in the reuptake of GABA by astrocytes which implies an increase in the GABA concentration in the extracellular space. We have shown (Proposition 1) that the excitability threshold $p_{\rm SNIC}$ increases linearly with the feedback term value depending on the GABA concentration in the extracellular space. Hence, such an astrocyte deficiency simulated in the model results in a decrease in the frequency of neural activity, which is consistent with the biological knowledge \cite{Brickley_2012, Luscher_2011}. The second result concerns the neuronal response to a deficiency in the reuptake of Glutamate by the astrocytes. In this case, the neural activity may either be reduced or enhanced or, alternatively, may experience a transient of high activity before stabilizing around a new activity state with a frequency close to the nominal one ({\it i.e.} before induction of astrocyte deficiency). We have characterized (Proposition 2 and Corollary 1) the relationships between parameters of the model for predicting the neuronal response. It is worth noticing that it is possible to calculate explicitly the SNIC bifurcation for the uncoupled NMM (see \cite{Garnier_2015}), but not for the bilaterally coupled model embedding the glial compartment. Hence, we have expressed the question of characterizing the variations in the neuronal excitability as an optimization problem under an equality constraint resulting from the implicit characterization of the saddle-node bifurcation.

The two model parameters $m_{\rm Glup}$ and $m_{\rm Glui}$ involved in Proposition 2 and Corollary 1 represent the maximal strengths of the Glutamate concentration impacts on the excitability of the pyramidal cell and the interneuron populations respectively. Our model-based study shows that the neural compartment may ``resist'' to the impact of glutamate excess in the extracellular space only if the value of $\frac{m_{\rm glup}}{m_{\rm glui}}$ lies in a positive interval. For small values ($\frac{m_{\rm glup}}{m_{\rm glui}}<I_1$), the neural activity frequency is lowered, as shown in Figure \ref{mglup2}, while, for high values ($\frac{m_{\rm glup}}{m_{\rm glui}}>I_2$), the activity frequency increases drastically and permanently, as shown in Figure \ref{mglup3}. For intermediary values, the neural activity recovers after a high frequency transient to a comparable mode despite the high values of extracellular Glutamate and GABA concentrations. Note that $I_1$ and $I_2$ explicitly depend on the coupling strength modulating the inhibition of the pyramidal cells activity by the interneurons ({\it i.e.} parameter $C_4$). On the other hand, only $I_2$ depends on $m_{\rm Glui}$. Hence, one might interpret this later dependency as a balance between the Glutamate-related glial feedbacks upon pyramidal cells and interneuron that the system should fulfill for being able to recover from a dysfunction in the astrocyte activity and avoid hyper-excitable behaviors.

An interesting fact should be noticed for future experimental investigation based on this study: for $\frac{m_{\rm glup}}{m_{\rm glui}}>I_1$, the astrocyte deficiency in capturing Glutamate induces, in addition to the increase in the extracellular Glutamate concentration, an increase in the GABA concentration, even greater than the one resulting from a default in GABA capture. Dynamically, this increase is necessary so that the system could reach an alternative state and recover a low frequency activity, but may not be sufficient if the glial feedback on the pyramidal cells and the interneurons is unbalanced, {\it i.e.} for to large values of $\frac{m_{\rm glup}}{m_{\rm glui}}$. The GABA increase induced by an excess of extracellular Glutamate and the corresponding neuronal response form an experimental benchmark to identify the mechanism of the transition through hyperexcitability and potential neuronal recovery.

Future theoretical works will extend the analysis of the local neuro-glial dynamical interactions by studying the model response when the neuronal compartment undergoes a drastic change in its bifurcation structure induced by a deficiency in the astrocyte compartment dynamics. As a matter of fact, for other parameter values corresponding to different inter- and intra-population connectivity strengths, the neural compartment may generate other types of time series than Noise-Induced Spiking (NIS) among those identified in \cite{Garnier_2015}. Therefore, the present study can be continued to study the dynamical mechanisms underlying successive transitions through various activity regimes. Another perspective will consist to consider two coupled neuro-glial models, the input $p(t)$ in each one being replaced by the output of the other one. The extension of the results presented in this article to such a system may provide useful methodological tools to tackle the hectic question of neuro-glial network modeling.

\subsubsection*{Acknowledgements}
This work was performed within the Labex SMART (ANR-11-LABX-65) supported by French state funds managed by the ANR within the Investissements d'Avenir program under reference ANR-11-IDEX-0004-02. We thank H. Berry for helpful comments on this work. We also thank S. Blanchard and F. Wendling for allowing us to use the glial model they submitted.

\end{document}